\documentclass[journal,transmag]{IEEEtran}

\usepackage{fancybox}
\usepackage{fancyhdr}

\usepackage[framed,autolinebreaks,useliterate]{mcode}
\usepackage{mathtools}
\usepackage{nicematrix}
\usepackage{float}
\usepackage{graphicx} 
\usepackage[colorinlistoftodos]{todonotes}
\graphicspath{ {images/} }

\usepackage{color, colortbl}
\usepackage{multirow}

\usepackage[linesnumbered,ruled,vlined]{algorithm2e} 

\usepackage[utf8]{inputenc}
\usepackage[center]{caption}
\usepackage[english]{babel}
\addto\captionsenglish{}
\usepackage{amsthm,amssymb}

\usepackage{optidef}
\usepackage{mathrsfs}

\usepackage{xcolor}

\theoremstyle{theorem}
\newtheorem{theorem}{Theorem}
\newtheorem{proposition}{Proposition}
\newtheorem{lemma}{Lemma}
\newtheorem{corollary}{Corollary}

\theoremstyle{definition}
\newtheorem{example}{Example}

\newtheorem{definition}{Definition}
\newtheorem{claim}{Claim}
\newtheorem{remark}{Remark}

\newcommand\underrel[2]{\mathrel{\mathop{#2}\limits_{#1}}}

\DeclareMathOperator{\argmax}{arg max}

\DeclareMathOperator{\bin}{bin}
\DeclareMathOperator{\supp}{supp}

\DeclareMathOperator{\w}{w}
\DeclareMathOperator{\dist}{d}

\newcommand{\wm}{w_{\min}}
\newcommand{\Awm}{A_{w_{\min}}(\I)}
\newcommand{\Aiw}{A_{i,w}(\I)}
\newcommand{\Aiwm}{A_{i,w_{\min}}(\I)}

\newcommand{\Ki}{\mathcal{K}_i}

\newcommand{\I}{\mathcal{I}}
\newcommand{\IRM}{\mathcal{I}_{\text{RM}}}
\renewcommand{\H}{\mathcal{H}}
\renewcommand{\S}{\mathcal{S}}
\newcommand{\Si}{\mathcal{S}_i}
\newcommand{\Sj}{\mathcal{S}_j}
\newcommand{\Sm}{\mathcal{S}_m}
\newcommand{\Sc}{\mathcal{S}_c}
\newcommand{\Sf}{\mathcal{S}_f}
\newcommand{\T}{\mathcal{T}}
\newcommand{\Ti}{\mathcal{T}_i}
\newcommand{\mJp}{m_{\mathcal{J}'}}
\newcommand{\JsS}{\mathcal{J}^*(\mathcal{S})}

\newcommand{\J}{\mathcal{J}}
\newcommand{\Jp}{\mathcal{J'}}
\newcommand{\F}{\mathcal{F}}
\newcommand{\M}{\mathcal{M}}
\newcommand{\MJ}{\mathcal{M}(\mathcal{J})}
\newcommand{\C}{\mathcal{C}}
\newcommand{\Ci}{\mathcal{C}_i}
\newcommand{\CI}{\mathcal{C}(\mathcal{I})}
\newcommand{\CIp}{\mathcal{C}(\mathcal{I}')}
\newcommand{\CiI}{\mathcal{C}_i(\mathcal{I})}
\newcommand{\G}{\mathcal{G}}
\newcommand{\B}{\mathcal{B}(\mathcal{I})}
\newcommand{\Bc}{\mathcal{B}(\mathcal{I}^c)}
\newcommand{\Bs}{\mathcal{B}^*(\mathcal{I}^c)}
\newcommand{\D}{\mathcal{D}}
\newcommand{\Dj}{\mathcal{D}_j}
\newcommand{\E}{\mathcal{E}}

\newcommand{\Ej}{\mathcal{E}_j}
\newcommand{\K}{\mathcal{K}}
\newcommand{\Q}{\mathcal{Q}}
\newcommand{\R}{\mathcal{R}}
\renewcommand{\O}{\mathcal{O}}

\newcommand{\JJ}{\mathfrak{J}(\mathcal{J})}
\newcommand{\SmJp}{\mathcal{S}_{m_{\mathcal{J}'}}}
\newcommand{\dm}{d_{\min}}
\newcommand{\Adm}{A_{d_{\min}}}

\newcommand{\bgi}{\mathbf{g}_i}
\newcommand{\bgj}{\mathbf{g}_j}
\newcommand{\bgm}{\mathbf{g}_m}
\newcommand{\bgf}{\mathbf{g}_f}
\newcommand{\bgh}{\mathbf{g}_h}
\newcommand{\bg}{\mathbf{g}}
\newcommand{\bGN}{\boldsymbol{G}_N}
\newcommand{\bc}{\boldsymbol{c}}
\newcommand{\bb}{\boldsymbol{b}}

\newcommand{\bv}{\boldsymbol{v}}
\newcommand{\bx}{\boldsymbol{x}}
\newcommand{\bG}{\boldsymbol{G}}
\newcommand{\bA}{\boldsymbol{A}}

\newcommand{\ff}{\mathbb{F}}
\newcommand{\ft}{\mathbb{F}_2}
\newcommand{\fq}{\mathbb{F}_q}

\newcommand{\cC}{{\mathscr{C}}}

\newcommand{\hg}[1]{\textcolor{black}{#1}}

\IEEEoverridecommandlockouts                              

\title{
On the Formation of Min-weight Codewords of Polar/PAC Codes and Its Applications
}

\author{
\IEEEauthorblockN{Mohammad Rowshan, {\em Member, IEEE}, Son Hoang Dau, {\em Member, IEEE}, and Emanuele Viterbo, {\em Fellow, IEEE}}
\thanks{Mohammad Rowshan is currently with the  School of Electrical Engineering and Telecommunications, The Univerity of New South Wales (UNSW), Sydney, NSW 2052, Australia. E-mail: m.rowshan@unsw.edu.au. This research was carried out during his Ph.D. program at Monash University.}
\thanks{Emanuele Viterbo is with the  Department of Electrical and Computer Systems Engineering (ECSE), Monash University, Melbourne, VIC 3800, Australia. E-mail: emanuele.viterbo@monash.edu.}
\thanks{Son Hoang Dau is with the School of Computing Technologies, RMIT University, Melbourne, VIC 3000, Australia. E-mail: sonhoang.dau@rmit.edu.au.}
\thanks{This work was supported by the Australian Research Council under Discovery Project ARC DP200100731 and DECRA Project DE180100768.}
\thanks{This paper was presented in part at the 2022 IEEE Information Theory Workshop (ITW), Mumbai, India \cite{rowshan2022errcoeff}.}
}

\begin{document}

\maketitle
\thispagestyle{empty}
\pagestyle{empty}
\thispagestyle{fancy}
\lhead{\centering This paper, previously titled ``{\color{blue}Error Coefficient-reduced Polar/PAC Codes}",  will appear in {\em IEEE Trans. on Inf. Theory}. } 
\cfoot{}

\begin{abstract}
Minimum weight codewords play a crucial role in the error correction performance of a linear block code. In this work, we establish an explicit construction for these codewords of polar codes as a sum of the generator matrix rows, which can then be used as a foundation for two applications. 
In the first application, we obtain a lower bound for the number of minimum-weight codewords (a.k.a. the error coefficient), which matches the exact number established previously in the literature.
In the second application, we derive a novel method that modifies the information set (a.k.a. rate profile) of polar codes and PAC codes in order to reduce the error coefficient, hence improving their performance.
More specifically, by analyzing the structure of minimum-weight codewords of polar codes (as special sums of the rows in the polar transform matrix), we can identify rows (corresponding to \textit{information} bits) that contribute the most to the formation of such codewords and then replace them with other rows (corresponding to \textit{frozen} bits) that bring in few minimum-weight codewords.
A similar process can also be applied to PAC codes.
Our approach deviates from the traditional constructions of polar codes, which mostly focus on the reliability of the sub-channels, by taking into account another important factor - the weight distribution. 
Extensive numerical results show that the modified codes outperform PAC codes and CRC-Polar codes at the practical block error rate of $10^{-2}$-$10^{-3}$.
\end{abstract}

\begin{IEEEkeywords}
Polarization-adjusted convolutional codes, PAC Codes, polar codes, minimum Hamming distance, weight distribution, list decoding, code construction, rate profile.
\end{IEEEkeywords}

\section{Introduction}
\label{sec:intro}
Polar codes \cite{arikan} are the first class of constructive channel codes that was proven to achieve the symmetric (Shannon) capacity of a binary-input discrete memoryless channel (BI-DMC) using a low-complexity successive cancellation (SC) decoder. 
However, the error correction performance of polar codes under SC decoding is not competitive. To address this issue, successive cancelation list (SCL) decoding was proposed in \cite{tal} which yields an error correction performance comparable to maximum-likelihood (ML) decoding at high SNR. Further improvement was obtained by concatenation of polar codes and cyclic redundancy check (CRC) bits \cite{tal} or parity check (PC) bits \cite{trifonov2,zhang}, and by convolutional pre-transformation, a.k.a. polarization-adjusted convolutional (PAC) codes \cite{arikan2}. 

The error correction performance of linear codes under ML decoding can be estimated by the Union bound \cite[Sect. 10.1]{lin_costello} based on the weight distribution. As the truncated Union bound, in particular at high SNR regimes, suggests, the number of minimum Hamming weight codewords (a.k.a. error coefficient) has the largest contribution to the calculation of this bound. Given the importance of the number of minimum-weight codewords, several attempts pursuing the enumeration of weight distribution, and in particular the minimum-weight codewords of polar codes, have been undertaken in the past. 

In \cite{liu_analys}, the authors proposed sending the all-zero codeword over a channel with low noise, or receiving at very high SNR, and counting the re-encoded messages with certain weights at the output of a successive cancellation list decoder with a very large list size. The method presented in \cite{valipour} suggests efficient computation of a probabilistic weight distribution expression.  
In \cite{bardet,bardet_arxiv}, a closed-form expression was proposed for the enumeration of
min-weight codewords of decreasing monomial codes, a large family of codes that includes polar codes and Reed-Muller codes. This work was recently extended in \cite{rowshan23closed} to the structure and enumeration of weights less than twice the minimum weight, in particular 1.5 times the minimum weight for polar codes. 
The authors in \cite{zhang_prob} proposed a way to obtain an approximate distance spectrum of polar codes with long lengths using the spectrum of short codes and a probabilistic assumption on the appearance of ones in codewords. Based on the weight distribution of $|u|u+v|$ constructed codes in \cite{fossorier}, the weight distribution of the words generated by the polar transform was found recursively in \cite{polyanskaya}. Note that this work does not count the codewords of a specific code where a subset of the rows of the polar transform is frozen, that is, it is not involved in the codeword formation. This shortcoming was addressed in \cite{yao} by proposing a recursive algorithm that counts all codewords from polar codes with any weight based on a specific definition of cosets. The authors of \cite{yao} also exploited the properties of monomial codes from \cite{bardet} to reduce the complexity of the proposed algorithm. Nevertheless, their algorithm cannot be used for medium and long block lengths. 

From a different perspective, the error coefficient of a code depends on the code construction. Polar codes are constructed by selecting good synthetic channels based on the reliability of the sub-channels in the polarized vector channel. Note that the vector channel is obtained from combining independent channels recursively, which results in polarized sub-channels. Bad synthetic channels are used for the transmission of known values (usually 0). The mapping of information bits to good sub-channels is performed based on a rate profile. Good sub-channels are selected based on various methods for the evaluation of sub-channels' reliability. In \cite{arikan}, a method based on the evolution of the Bhattacharyya parameters was used, and the Bhattacharyya parameters evolved through the channel combining process were the reliability metrics for binary erasure channels (BEC). This method does not provide an accurate reliability metric for low-reliability sub-channels under additive white Gaussian noise (AWGN) channels. Density evolution (DE) was proposed in \cite{mori} for a more accurate reliability evaluation. However, it suffers from excessive complexity. To reduce the complexity of DE, a method based on the upper bound and the lower bound on the error probability of the sub-channels was proposed in \cite{tal2}.  To further reduce the computational complexity of DE, the Gaussian approximation (GA) to evolve the mean log-likelihood ratios (LLR) throughout the decoding process in \cite{trifonov} which was based on \cite{chung}. There are also SNR-independent low-complexity methods for reliability evaluation. In \cite{schurch}, a partial ordering of sub-channels was proposed based on their indices. A method for ordering all the sub-channels was suggested in \cite{he} based on the binary expansion of the sub-channels indices. This method is known as the polarization weight (PW) method.

The aforementioned code construction methods estimate the reliability of sub-channels with different precision levels and various levels of computational complexity. However, selecting only good channels, i.e. sub-channels with the highest reliability, may result in poor weight distribution. Hence, to obtain a good error correction performance, one may not rely only on the reliability of the individual sub-channels. In \cite{rowshan-how}, an approach was proposed for constructing codes for list decoding in which the probability of elimination of the correct sequence in different sub-blocks of a code is balanced. In this scheme, a code obtained from traditional code construction methods is modified. A different method was suggested in \cite{trifonov-rand} to construct randomized polar subcodes that rely on the explicit enumeration of low-weight codewords in a polar code and the construction of dynamic freezing constraints (DFC) to eliminate most of these codewords. The numerical results have shown a significant performance gain for 1 kb code-length in high-SNR regimes and block error rate (BLER) below $10^{-4}$ and $10^{-5}$.
However, the DFCs are optimized and compared to non-optimized polar subcodes and CRC-polar codes.
Some other approaches such as in \cite{coskun,miloslavskaya} were proposed for designing improved polar-like codes for list decoding as well, although they do not provide explicit procedures for constructing a code. 

In this work, we first study the properties of the polar transform, a matrix resulting from the Kronecker power of the 2x2 binary Hadamard matrix, and characterize the rows involved in the generation of minimum-weight codewords. 
Although our characterization rediscovers a known formula for the number of minimum-weight codewords of polar codes developed in Bardet \textit{et al.}~\cite{bardet, bardet_arxiv}, it offers a different perspective that facilitates a novel approach for code modifications. 
Based on this, we propose a simple, low-complexity, and explicit method to modify polar and PAC codes to reduce the error coefficient $\Adm$, i.e. the number of minimum-weight codewords. Our method seeks to balance the competing effects of reducing the error coefficient and using some less reliable sub-channels to improve the error performance. 
The codes designed by this approach outperform polar codes and PAC codes in terms of block error rate (BLER) in certain regimes. 
More specifically, as demonstrated by the numerical results, the proposed codes have an edge at low and medium SNR regimes (where the gain is usually harder to achieve), and the BLER of $10^{-2}$-$10^{-3}$ over polar codes and their well-known variants. This BLER level is commonly used in many use cases, except in ultra-reliability low-latency communications (URLLC). Furthermore, we compare our results with the BLER lower bound for finite-length codes as a reference. In summary, our contributions are given below. 
\begin{itemize}
 \item We establish a construction of minimum-weight codewords in a polar code as a sum of a row~$i$ of minimum weight $\wm$, a set of \textit{core} rows (rows that are at distance $\wm$ from the row $i$), and a set of \textit{balancing} rows, which brings the weight of the sum back to $\wm$. This construction (see Theorem~1) immediately leads to a lower bound on the number of minimum-weight codewords in a polar code. 
\item We provide an analysis of error coefficient improvement in the convolutional precoding process in PAC coding.
\item Based on our new understanding of the structures of minimum-weight codewords, we develop a code modification procedure to improve the error coefficient of polar and PAC codes, targeting the low SNR regimes and BLER of $10^{-2}$-$10^{-3}$.
\end{itemize}

\textbf{Paper Outline:}
The rest of the paper is organized as follows. 
We provide in Section II basic concepts and notations in coding theory, as well as introduce Reed-Muller codes and polar codes and the relationship between them. 
In Section III, we study the special formation of minimum-weight codewords in polar codes. 
In Sections IV and V, leveraging the new insight regarding such formation, we propose a method to improve the error coefficient of polar codes by carefully modifying existing codes. 
We discuss in Section VI the impact of precoding on the error coefficient of existing polar codes and modified ones. In Section VII, we analyze the trade-off between the improvement of the error coefficient and the overall reliability at different SNR regimes.
The numerical results of the proposed construction are provided in Section VIII, while concluding remarks are given in Section IX. The Appendix contains several parts, which provide a MATLAB script for the enumeration of minimum-weight codewords (Appendix A), the relation between the error coefficient and block error probability (Appendix B), fundamental properties of polar transform (Appendix C), and a full proof for Theorem~\ref{thm:decomposition}, which is about the formation of the minimum-weight codewords in polar codes (Appendix D).

\section{Preliminaries} 
\label{sec:pre}

\subsection{Basic Concepts in Coding Theory}
\label{subsec:basic}

We denote by $\fq$ the finite field with $q$ elements. In this work we concentrate only on binary codes, that is, $q = 2$. 
The cardinality of a set is denoted by $|\cdot|$. The notation $\mathcal{V}_i^j$ represents a vector $V_i,V_{i+1},\cdots,V_j$.
We define in the following standard notions from coding theory (for instance, see \cite{lin_costello}).
The \emph{support} of a vector $\bc = (c_0,\ldots,c_{N-1}) \in \fq^N$ is the set of indices where $\bc$
has a non-zero coordinate, that is, $\supp(\bc) \triangleq \{i \in [0,N-1] \colon c_i \neq 0\}$. 
The (Hamming) \emph{weight} of a vector $\bc \in \fq^N$ is $\w(\bc)\triangleq |\supp(\bc)|$, which is the number of non-zero coordinates of $\bc$. 
For the two vectors $\bc = (c_0, c_1, \ldots, c_{N-1})$ and $\bc' = (c'_0, c'_1, \ldots, c'_{N-1})$ in $\fq^N$, 
the (Hamming) \emph{distance} between $\bc$ and $\bc'$ is defined to be the number of coordinates 
where $\bc$ and $\bc'$ differ, namely, 
\[
\dist(\bc,\bc') = |\{i \in [0,N-1] \colon c_i \ne c'_i\}|. 
\]
A $K$-dimensional subspace $\cC$ of $\fq^N$ is called a linear $(N,K,d)_q$ \emph{code} over $\fq$ 
if the minimum distance of $\C$, 
\[
\dist(\C) \triangleq \min_{\bc,\bc' \in \C, \bc \neq \bc'} \dist(\bc,\bc'), 
\]
is equal to $d$. Sometimes we use the notation $(N,K,d)$ or just $(N,K)$ for brevity. We refer to $N$ and $K$ as the \emph{length} and the \emph{dimension} of the code. The vectors in $\C$ are called \emph{codewords}. It is easy to see that the minimum-weight of a no-nzero codeword in a linear code $\C$ is equal to its minimum distance $\dist(\C)$. 
A \emph{generator matrix} $\G$ of an $(N,K)_q$ code $\cC$ is a $K \times N$ matrix in $\fq^{K\times N}$ whose rows are $\fq$-linearly independent codewords of $\C$. Then $\C = \{\bv \G \colon \bv \in \fq^K\}$.
We denote the number of codewords in $\cC$ with weight $w$ by $A_{w}(\cC)$. For brevity, we may drop $\cC$ and simply write $A_w$.

Let $[\ell,u]$ denote the range $\{\ell,\ell+1,\ldots,u\}$.
The binary representation of $i \in [0,2^n-1]$ is defined as  $\bin(i)=i_{n-1}...i_1i_0$, where $i_0$ is the least significant bit, that is $i = \sum_{a=0}^{n-1}i_a 2^a$. 
For $i \in [0,2^n-1]$, let $\Si$ denote the support of $\bin(i)$, that is, 
\[
\Si\triangleq\supp(\bin(i)) = \{a\in [0,n-1]\colon i_a = 1\}\subseteq [0,n-1].
\]
This is an important notation that we will use throughout this work.  
For instance, for $i=6=(00110)_2$, $\Si=\{1,2\}$. 
Note that the Hamming weight of $\bin(i)$ is $\w(\bin(i)) = |\Si|$.
We will use interchangeably $i\in [0,2^n-1]$ and $\Si$ as the index subscript of a codeword coordinate, i.e. $c_i=c_{\Si}$.
For example, when $n = 5$, we may use $\Si=\{1,3\}$ to refer to the index $i = 10$, which has $\bin(i)=01010$, and write $c_{\{1,3\}}$ instead of $c_{10}$.
We also define $\Si$'s complement $\mathcal{T}_i$ as 
$\Ti\triangleq [0,n-1]\setminus\Si$.
For instance, when $n = 5$ and $i=10$, we have $\Ti=\{0,2,4\}$.

\subsection{Reed-Muller Codes and Polar Codes}
\label{subsec:RMPolar}
Reed-Muller (RM) codes and polar codes of length $N=2^n$ are constructed based on the $n$-th Kronecker power of binary Walsh-Hadamard matrix  
$\mathbf{G}_2 = 
{\footnotesize \begin{bmatrix}
1 & 0 \\
1 & 1
\end{bmatrix} }$, that is, $\bGN=\mathbf{G}_2^{\otimes n}$, which is referred to as {\em polar transform} throughout this paper. We denote polar transform by rows as
\begin{equation}
    \bGN= \begin{bmatrix}
        \bg_0\\
        \bg_1\\
        \vdots\\
        \bg_{N-1}
    \end{bmatrix}.
\end{equation}
A generator matrix of RM code or polar code is formed by selecting a set of rows of $\bGN$. We use $\I$ to denote the set of indices of these rows and $\CI$ to denote the linear code generated by the set of rows of $\bGN$ indexed by $\I$. Note that $\I \subseteq [0,N-1]=[0,2^n-1]$. We describe below how to select the information sets $\I$ for RM and polar codes, respectively.

\textbf{Reed-Muller Codes}. The generator matrix of RM code of length $2^n$ and order $r$, denoted RM$(r,n)$, is formed by the set of all rows $\bgi, i\in[0,N-1]$, of weight $\w(\bgi)\geq 2^{n-r}$, which is the minimum-weight $\wm$ of the code. Therefore, the information set $\IRM$ of RM$(r,n)$ is created as follows. 
\[
\IRM = \{i\in [0,2^n-1]\colon \w(\bgi)\geq 2^{n-r}\}.
\]
The dimension of RM$(r,n)$ is $K=|\IRM|=\sum_{\ell=0}^r \binom{n}{\ell}$. 
The concept of order $r$ in the RM($r, n$) code comes from the wedge products of $N$-tuple $\mathbf{v}_{i+1}=\mathbf{g}_{N-2^i-1}, i\in[0,n-1]$ up to degree $r$, where $\bgj$ is the $j$-th rows of $\bGN$. By default, $\mathbf{v}_0 = \mathbf{g}_{N-1} = [1\;1\;...\;1]$. For instance, when $n=3$ we obtain 
\begin{gather*}
    \mathbf{v}_0 = \mathbf{g}_7 = [1\;1\;1\;1\;1\;1\;1\;1],\\    \mathbf{v}_1 = \mathbf{g}_6 = [1\;0\;1\;0\;1\;0\;1\;0],\\
    \mathbf{v}_2 = \mathbf{g}_5 = [1\;1\;0\;0\;1\;1\;0\;0],\\
    \mathbf{v}_3 = \mathbf{g}_3 = [1\;1\;1\;1\;0\;0\;0\;0].
\end{gather*}
The vectors $\mathbf{v}_i, i\in[0,3]$ in the example above form the generator matrix of RM(1,3). As an example for order 2, the generator matrix for RM$(2,3)$ is given by
\begin{equation}\label{eq:G_RM2_3}
\mathbf{G}_{\text{RM}(2,3)}\!=\! 
\begin{bmatrix}
\mathbf{v}_2{\tiny\wedge}\mathbf{v}_3 \\
\mathbf{v}_1{\tiny\wedge}\mathbf{v}_3 \\
\mathbf{v}_3 \\
\mathbf{v}_1{\tiny\wedge}\mathbf{v}_2 \\
\mathbf{v}_2 \\
\mathbf{v}_1 \\
\mathbf{v}_0 \\
\end{bmatrix}\!=\!
\begin{bmatrix}
1 & 1 & 0 & 0 & 0 & 0 & 0 & 0\\
1 & 0 & 1 & 0 & 0 & 0 & 0 & 0\\
1 & 1 & 1 & 1 & 0 & 0 & 0 & 0\\
1 & 0 & 0 & 0 & 1 & 0 & 0 & 0\\
1 & 1 & 0 & 0 & 1 & 1 & 0 & 0\\
1 & 0 & 1 & 0 & 1 & 0 & 1 & 0\\
1 & 1 & 1 & 1 & 1 & 1 & 1 & 1
\end{bmatrix}
\end{equation}
which has rows with minimum Hamming weight $2^{n-r}=2^{3-2}=2$. One can observe that the generator matrix of RM$(n,n)$ is $\bGN$. In the example above, only $\mathbf{g}_0$ of $\mathbf{G}_8$ which has weight 1 is not included in \eqref{eq:G_RM2_3}. 

\textbf{Polar Codes}. The characterisation of the information set $\I$ for polar codes is more cumbersome, relying on the concept of \textit{bit-channel reliability}. 
We discuss this in detail in the next few paragraphs.

The key idea of polar codes of length $N=2^n$ lies in using a polarization transformation that converts $N$ identical and independent copies of any given binary-input discrete memoryless channel (BI-DMC) $W$  into $N$ synthetic channels $\{W^{(i)}_N, 0\leq i\leq N-1\}$ which are either better or worse than the original channel $W$ 
\cite{arikan}. 
We define $W_N(y^{N-1}_0|u^{N-1}_0)=W_N(y^{N-1}_0|u^{N-1}_0\bGN)$ as the polarized vector channel from the transmitted bits $u^{N-1}_0$ where $y^{N-1}_0$ are the received signals from the $N$ copies of the physical channel $W$. The bit-channel $W^{(i)}_N, i\in[0,N-1]$ is implicitly defined as
\begin{equation*} 
W_{N}^{(i)}\left ({y_{1}^{N}, u_{1}^{i-1}|u_{i}}\right ) = \sum _{u_{i+1}^{N}} \frac {1}{2^{N-i}}W_{N}\left ({y_{1}^{N}|u_{1}^{N}}\right )\!. 
\end{equation*}
The channel polarization theorem \cite{arikan} states that the symmetric capacity of the bit-channel $W^{(i)}_N$, denoted $I(W^{(i)}_N)$, converges to either 0 or 1 as $N$ approaches infinity. It can also be shown that the fraction of the channels that become perfect converges to the capacity of the original channel $W$, i.e., $I(W)$, meaning that polar codes are {\em capacity achieving} while the fraction of extremely bad channels approaches to ($1-I(W)$). 

Hence, a polar code of length $N=2^n$ is constructed by selecting a set $\I$ of indices $i\in[0,N-1]$ with the highest $I(W^{(i)}_N)$. The indices in $\I$ are dedicated to information bits, while the rest of the bit-channels with indices in $\mathcal{I}^c \triangleq [0,N-1]\setminus \I$ are used to transmit a known value, `0' by default, which are called \emph{frozen bits}. 
Regardless of the method we use for forming the set $\mathcal{I}$ for a polar code, the bit-channels with indices in the set $\mathcal{I}$ must be more reliable than any bit-channels in $\mathcal{I}^c$. The notation $W^{(i)}_N\preceq W^{(j)}_N$ is used to say that the bit-channel $j$ is more reliable than bit-channel $i$. 
In summary, a polar code can be defined by any set $\mathcal{I} \subseteq [0,N-1]$ satisfying $W^{(i)}_N\preceq W^{(j)}_N$ for every $j\in\mathcal{I}, i\in\mathcal{I}^c$. Such a code has dimension $K = |\I|$.

\subsection{Partial Order Property and a Generalization of Reed-Muller and Polar Codes}
\label{subsec:POP}

In the first part of this work, we 
identify the minimum-weight codewords for a more general family of linear codes $\CI$ that includes both RM codes and polar codes as special cases. 
This family of codes is defined based on the partial orders introduced in the literature of polar codes (\cite{mori, schurch, wu}), which are based on the binary representations of the bit-channel indices and conveniently abstracts away the cumbersome notion of bit-channel reliability. We first define in Definition~\ref{def:PO} these partial orders, combined as a single partial order, and then the so-called \textit{Partial Order Property} that the information set $\I$ of these codes needs to satisfy in Definition~\ref{def:POP}.
We came to know when writing that the same family of code had been also investigated in the previous work by Bardet \textit{et al.}~\cite{bardet, bardet_arxiv} under the name of \textit{decreasing monomial codes}. 

\begin{definition}[Partial Order]
\label{def:PO}
Given $i,j\in [0,2^n-1]$, we denote $i \preceq j$ or $j \succeq i$ if they satisfy one of the following conditions:
\begin{itemize}
    \item $\Si \subseteq \Sj$, 
    \item $\Sj = (\Si\setminus \{a\})\cup \{b\}$ for some $a\in \Si$, $b \notin \Si$ and $a < b$ (i.e., $\bin(j)$ is obtained from $\bin(i)$ by swapping a `1' in  $\bin(i)$ and a `0' at a higher index),
    \item there exists $k\in [0,2^n-1]$ satisfying $i\preceq k$ and $k\preceq j$,
\end{itemize}
where $\Si \triangleq \supp(\bin(i))\subseteq [0,n-1]$, which consists of the indices where $i$ has a `1' in its binary representation.
Note that $i \preceq j$ implies that $i \leq j$ but not vice versa.
\end{definition}

It is straightforward to verify that Definition~\ref{def:PO} defines a partial order, i.e. a binary relation on the set $[0,2^n-1]$ satisfying reflexivity, antisymmetry, and transitivity.

It turns out that the relative reliability of some pairs of the bit-channels with indices in $[0,N-1]$ can be determined using the partial order defined in Definition~\ref{def:PO} as follows.

\begin{proposition}[\cite{mori, schurch, wu, bardet, bardet_arxiv}]
\label{pro:PO}
If $j \succeq i$ then the bit-channel $W^{(j)}_N$ is more reliable than the bit-channel $W^{(i)}_N$.

\end{proposition}

\begin{definition}[Partial Order Property]
\label{def:POP}
A set $\I\subseteq [0,2^n-1]$ is said to satisfy the \textit{Partial Order Property} if $i \not\preceq i^c$ for every $i \in \I$ and $i^c\in \I^c$. In other words, none of the indices in $\I$ is smaller than or equal to another index in $\I^c$ according to the partial order defined in Definition~\ref{def:PO}. Equivalently, for every $i\in \I$ and $j\in [0,N-1]$, if $j \succeq i$ then $j\in \I$.
\end{definition}

\begin{corollary}
\label{cr:POP}
The information sets $\I$ of Reed-Mular codes and polar codes satisfy the Partial Order Property. 
\end{corollary}
\begin{proof}
For the $\text{RM}(r,n)$, we have 
\begin{equation}\label{eq:I_RM}
\begin{split}
    \IRM &= \{i\in [0,2^n-1]\colon \w(\bgi)\geq 2^{n-r}\}\\
    &=\{i\in [0,2^n-1]\colon |\Si|\geq n-r\},
\end{split}
\end{equation}
where the second equality is due to Corollary~\ref{cr:weight} (Appendix~\ref{app:polar_transform}), 
which states that $\w(\bgi) = 2^{|\Si|}$. Clearly, if $i^c \in \I^c$ then $|\S_{i^c}| < n-r \leq |\Si|$, which implies that $i \not\preceq i^c$. Therefore, RM codes satisfy the Partial Order Property.

For a polar code, its information set $\I$ must satisfy the condition that the bit-channel $W^{(i)}_N$ is more reliable than the bit-channel $W^{(i^c)}_N$ for every $i\in \I$ and $i^c \in \I^c$. By Proposition~\ref{pro:PO}, such $i$ and $i^c$ must satisfy $i \not\preceq i^c$. Therefore, the information set $\I$ of a polar code $\CI$ satisfies the Partial Order Property. 
\end{proof}

It is a simple fact that the linear codes in the general family we are considering are subcodes of RM codes.

Note that the selected rows in the polar codes of length $N = 2^n$ and minimum row weight $\wm= \min \w(\bgi), i\in\mathcal{I}$ are a subset of the rows of the generator matrix for RM$(r,n$) where $2^{n-r}=\wm$. As a result, any polar code is a subcode of some RM code with common minimum distance which results in  $\mathcal{I}\subseteq \mathcal{I}_{RM}$.

\section{The Formation of Minimum-Weight Codewords of Reed-Muller and Polar Codes}
\label{sec:decomposition}

\subsection{The Minimum-Weight Codewords Formation}

To determine the minimum-weight codewords of a RM code or a polar code $\C(\I)$ generated by a set of rows $\{\bgi \colon i \in \I\}$ of $\bGN$, our strategy is to partition the code into $|\I|$ disjoint cosets $\CiI=\bgi+\C(\I \setminus [0,i])$ of its subcodes $\C(\I \setminus [0,i])$ for $i\in \I$, and identify minimum-weight codewords in each of such cosets. 
We came to realize when writing this paper and its conference version that another approach based on permutation groups of Reed-Mular/polar codes had been proposed before by Bardet \textit{et al.}~\cite{bardet, bardet_arxiv}. Our work is a set-theoretic approach and achieves only the lower bound on the number of minimum-weight codewords of $\CI$, while both the (same) lower bound and a matching upper bound were established in~\cite{bardet, bardet_arxiv}. However, the formation of minimum-weight codewords proposed in our approach makes it more convenient to make a modification to the codes to reduce the number of minimum-weight codewords and achieve a better performance. We discuss the connection between our approach and that of Bardet \textit{et al.} in detail at the end of this section.

\begin{definition} 
\label{def:CI}
Given a set $\I \subseteq [0,N-1]$, we define the set of codewords $\CiI\subseteq \CI$ for each $i\in \I$ as follows. 
\begin{equation}
    \CiI \triangleq \left\{\bgi\oplus\bigoplus_{h\in \H} \bgh \colon \H \subseteq \I \setminus [0,i]\right\}\subseteq \C(\I).
\end{equation}
In other words, $\CiI$ is a coset of the subcode $\C(\I \setminus [0,i])$ of $\C(\I)$ generated by $\{\bgh\colon h \in \I\setminus [0,i]\}$ with the coset leader $\bgi$, where $\bgi$ is the $i$-th row of the polar transform $\bGN$.
It is clear that the sets $\CiI$, $i \in \I$, partition the code $\CI$. 

\end{definition}

\begin{lemma}
Let $\Awm$ denote the number of codewords of minimum-weight $\wm$ of the RM/polar code $\C(\I)$, and $\Aiw$ denote the number of codewords of weight $w$ in the coset $\CiI$, $i\in \I$. Then the following formula holds.
\begin{equation}
\Awm = \sum_{i \in \I\colon \w(\bgi) = \wm} \Aiwm.
\end{equation}
\end{lemma}
\begin{proof}
Since $\CiI$, $i \in \I$, partition the code $\CI$, we have
\[
\Awm = \sum_{i \in \I} \Aiwm = \sum_{i \in \I\colon \w(\bgi) = \wm} \Aiwm,
\]
where the second equality holds because due to Corollary~\ref{cor:geq_wi} (Appendix~\ref{app:polar_transform}), 
if $\w(\bgi)>\wm$ then all codewords in $\CiI$ have weights greater than $\wm$, that is, $\Aiwm = 0$. \qedhere
\end{proof}

We now define the set of indices $\Ki$, which plays an essential role in the formation of minimum-weight codewords in $\CiI$. If $\I$ satisfies the Partial Order Property then $|\Ki|$ is the number of minimum-weight codewords in $\CiI$ 
of the form $\bc = \bgi + \bgj$, $j \in \I\setminus [0,i]$.
Surprisingly, $\Ki$ also leads to the formation of all other minimum-weight codewords in $\CiI$ (see Theorem~\ref{thm:decomposition}).

\begin{definition}
\label{def:Ki}
For each index $i \in [0,N-1]$ we define
\[
\Ki \triangleq \{j \in [i+1,N-1]\colon \w(\bgj)\geq \w(\bgi+\bgj)=\w(\bgi)\}
\]
as the set of indices $j \in [i+1,N-1]$ so that $\bgj$ is at distance $\w(\bgi)$ away from $\bgi$ and has weight at least $\w(\bgi)$.
\end{definition}
We call the rows indexed by elements in $\Ki$ the \emph{core rows} of $i$. As we will see later, 
the core rows allow one to form all minimum-weight codewords in $\CiI$.  
The properties of $\Ki$ are listed in Lemma~\ref{lma:Ki}.
Note that the definition of $\Ki$ is code independent. However, thank to Lemma~\ref{lma:Ki}~c), if $\I$ satisfies the Partial Order Property and $i \in \I$ then $\Ki\subseteq \I$. Hence, $|\Ki|$ is indeed the number of minimum-weight codewords in $\CiI$ that are the sums of $\bgi$ and another row $\bgj$, $j \in \I\setminus [0,i]$.

\begin{lemma}
\label{lma:Ki}
The set $\Ki$ defined in Definition~\ref{def:Ki} satisfies the following properties.
\begin{enumerate}
    \item[a)] $\Ki = \{j \in [i+1,N-1]\colon |\Sj\setminus\Si|=1, |\Sj|=|\Si| \text{ or }|\Si|+1 \}$. 
    \item[b)] For every $j\in\Ki$, we have 
        \begin{equation}
        \Sj\cap\Si=
        \begin{cases}
            \Si, & \text{if } \w(\bgj)=2\w(\bgi),\\
            \Si\setminus\{k\}, k\in\Si, & \text{if } \w(\bgj)=\w(\bgi).\\
        \end{cases}
        \end{equation}
    \item[c)] If $\I \subseteq [0,N-1]$ satisfies the Partial Order Property then for every $i\in\mathcal{I}$, 
    we have $\mathcal{K}_i\subseteq\mathcal{I}$. 
    \item[d)] The size of $\Ki$ is (recalling that $\bin(i)=i_{n-1}...i_1i_0$)
    \begin{equation}\label{eq:Ki_size}
        |\mathcal{K}_i|=|\Ti|+\sum_{k\in \Si}\sum_{\ell>k} \bar{i}_\ell,
    \end{equation}
    where $\Ti\triangleq [0,n-1]\setminus \Si$ and $\bar{i_\ell}\triangleq i_\ell\oplus 1$.
\end{enumerate}
\end{lemma}
\begin{proof}
The proof for each part is given below.
\begin{enumerate}
    \item[a)] According to Corollary~\ref{cr:weight} (Appendix~\ref{app:polar_transform}), $\w(\bgi+\bgj)=\w(\bgi)$ if and only if $|\Sj| = 1+|\Si\cap \Sj|$. Taking into account the condition that $\w(\bgj)\geq\w(\bgi)$, or equivalently, $|\Sj|\geq |\Si|$, we conclude that $j\in \Ki$ if and only if $|\Sj\setminus \Si| = 1$ and $|\Si|\leq |\Sj| \leq |\Si|+1$.

    \item[b)] Following properties of $\Ki$ in part (a), when $|\Sj|=|\Si|$, then according to Corollary \ref{cr:weight} we have $\w(\bgj)=\w(\bgi)$. In this case since $|\Sj\setminus\Si|=1$,  then $|\Sj\cap\Si|=|\Si|-1$ which implies that $\Sj\cap\Si=\Si\setminus\{k\}$ for some $k\in\Si$. Also, when $|\Sj|=|\Si|+1$, then we have $\w(\bgj)=2\w(\bgi)$. In this case since $|\Sj\setminus\Si|=1$, then $|\Sj\cap\Si|=|\Si|$ which implies that $\Sj\cap\Si=\Si$.
    \item[c)] According to Part b), if $j \in \Ki$ then either $\Sj \supseteq \Si$ or $\Sj$ is obtained from $\Si$ by replacing an index $k\in \Si$ with another index $h > k$ (because $j>i$). By Definition~\ref{def:PO}, $j \succeq i$. Since $\I$ satisfies the Partial Order Property, $j \in \I$. Therefore, $\Ki \subseteq \I$.    
    \item[d)] To count the elements of $\Ki$, we consider two cases in part (b) in addition to the condition $|\Sj\setminus\Si|=1$:
    \begin{itemize}
        \item If $\Sj\cap\Si=\Si$, then we count any $j$ where there exists some $\ell\in\Sj$ and $\ell\in\Ti$. That is, by \emph{addition} of one $\ell\in\Ti$ at a time to $\Si$, we can obtain all such $j$ rows. Thus, we have $|\Ti|$ such $j$ rows in total. 
        \item If $\Sj\cap\Si=\Si\setminus\{k\},k\in\Si$, then we count any $j$ where there exists some $\ell\in\Sj$ and $\ell\in\Ti$ in which $\ell$ is swapped with some $k\in\Si$ to retain $|\Sj|=|\Si|$. Since $j>i$, this swap should be \emph{left-swap} as the right-swap gives $j<i$.  Hence to count all such  $j$ rows, for every $k\in\Si$ we count all $\ell\in\Ti$ such that $\ell>k$. This operation can be implemented by $\sum_{k\in \Si}\sum_{\ell>k} \bar{i}_\ell$ where $\bar{i_\ell}=i_\ell\oplus 1$.\qedhere
    \end{itemize}
\end{enumerate}
\end{proof}

\begin{remark}\label{rmk:add_swap}
From the proof of Lemma \ref{lma:Ki} - part (c), 
note that the elements of $\Ki$ can be obtained by applying the \emph{addition} and \emph{left-swap} operations on $\bin(i)$: 
\begin{itemize}
    \item Addition: if we flip every `0' in $\bin(i)$ one at time, we get all $j\in\Ki$ which have weight $\w(\bgj)=2\w(\bgi)$.
    \item Left-swap: if we swap every `1' in $\bin(i)$ with every `0' on the left, one at time, we get all $j\in\Ki$ which have weight $\w(\bgj)=\w(\bgi)$.
\end{itemize}
\end{remark}

\begin{example} Suppose $i=(13)_{10}=(01101)_2$ for $\mathbf{G}_{32}$ where $n=5$. To find the set $\mathcal{K}_i$, we follow Remark \ref{rmk:add_swap}. First we find all $j>i$ with weight $2\w(\bgi)$ by addition operation. These rows  are $\{(01111)_2,(11101)_2\}=\{15,29\}\subset \mathcal{K}_i$. The size of this subset can be found even without listing them by $|\Ti|=n-|\Si|=5-\w(01101)=2$. We are actually counting the number of zero-value positions in $\bin(i)$. Then, we find all $j>i$ with weight $\w(\bgi)$ by left-swap operation over $\bin(i)$. These rows are $\{(01110)_2,(11100)_2,(11001)_2,(10101)_2\}=\{14,28,25,21\}\subset \mathcal{K}_i$. The size of this subset also can be found without listing them by $\big(5-\w(0110\underline{1})\big)+\big(3-\w(01\underline{1})\big)+\big(2-\w(0\underline{1})\big)=4$. We are actually counting the number of zero-value positions in $\bin(i)$ at the positions larger than $k$, positions $k$ are underlined. Hence, $|\mathcal{K}_i|=6$ and 
\begin{equation*}
    \Ki=\{14,15,21,25,28,29\}.
\end{equation*}
\end{example}

We show in Theorem~\ref{thm:decomposition} that if the information set $\I$ satisfies the Partial Order Property then the set $\Ki$, although defined to capture \textit{some} specific minimum-weight codewords in $\CiI$, allows us to identify \textit{all} minimum-weight codewords of $\CI$ lying in $\CiI$ for every $i \in \I$. \hg{Note that by its own right, the theorem only implies a lower bound on the number of minimum-weight codewords (see Corollary~\ref{trm:2}). However, given the work of Bardet \emph{et al.}~\cite{bardet}, we know that this bound is exact. 
The theorem applies to both RM and polar codes thanks to Corollary~\ref{cr:POP}.}

\begin{theorem}\label{thm:decomposition}
Suppose that $\I\subseteq [0,N-1]$ satisfies the Partial Order Property and $i\in \I$ is such that $\w(\bgi)=\wm$. Then for any set $\mathcal{J}\subseteq \mathcal{K}_i$, there exists a set $\M(\J) \subseteq \I\setminus \Ki$ such that
\begin{equation}\label{eq:decomposition}
\w\big(\bgi\oplus\underbrace{\bigoplus_{j\in\J}\bgj}_{\textnormal{core rows}} \oplus\underbrace{ \bigoplus_{m\in\M(\J)}\bgm}_{\textnormal{balancing rows}}\big) =  \wm.
\end{equation}
Moreover, such a set $\MJ$ can be constructed by the $\M$-Construction (see below).
Note that the rows in $\M(\J)$ are called balancing rows as their inclusion brings the weight of the sum down to $\wm$ if the sum of the coset leader and a subset of core rows has weight exceeding $\wm$.
\end{theorem}
\begin{proof}
The theorem is proved in Appendix~\ref{app:construction}.
\end{proof}



\textbf{$\M$-Construction.} Suppose that $\I\subseteq [0,N-1]$ satisfies the Partial Order Property and $i\in \I$ satisfying $\w(\bgi)=\wm$. For any $\varnothing \neq \J\subseteq \Ki$, we aim to construct a set $\M(\J) \subseteq \I\setminus (\Ki\cup [0,i])$ satisfying \eqref{eq:decomposition}\footnote{Note that while showing that $\MJ \subseteq \I\setminus [0,i]$ requires Lemma~\ref{lma:M_in_I} in Appendix~\ref{app:construction}, the fact that $\MJ \cap \Ki=\varnothing$ can be seen from the $\M$-Construction itself because $|\SmJp\setminus \Si|\geq 2$, and hence $\mJp$ doesn't satisfy Lemma~\ref{lma:Ki}~(a).}. First, let
\[
\JJ\triangleq \left\{\J'\subseteq \J\colon |\J'|\geq 2, \Sj\setminus \Si, j\in \J' \text{ are disjoint}\right\},
\]
noting that $|\Sj\setminus \Si|=1$ due to Lemma~\ref{lma:Ki}~a). 
Next, for every such $\J'\in\JJ$, let $\mJp\in [0,2^n-1]$ such that
\begin{equation}\label{eq:S_mJ}
\SmJp \triangleq \bigcup_{j\in \J'}\big(\Sj\setminus \Si\big) \cup\Big(\Si\cap\big(\bigcap_{j\in\J'}\Sj\big)\Big),
\end{equation}
The set $\MJ$ consists of all such $\mJp$ indices with odd multiplicities. More specifically, 
\begin{equation}\label{eq:set_M}
\MJ \triangleq \left\{h\in [0,N-1]\colon \left|\mathfrak{J}_{h}(\mathcal{J})\right| \text { is odd}\right\},
\end{equation}
where
\begin{equation}\label{eq:set_J_h}
\mathfrak{J}_{h}(\J)\triangleq\left\{\J' \in \JJ\colon \mJp=h\right\}.
\end{equation}

\begin{remark}
\hg{An equivalent way to define $\JJ$ and $\mathcal{S}_{m_{\J'}}$ in the $\M$-Construction is as follows. First, let
\begin{equation*}
\R=\bigcup_{j\in\mathcal{J}}\big(\Sj\setminus\Si\big) \subseteq \Ti \triangleq [0,n-1]\setminus \Si.
\end{equation*}
Then, we can verify that 
\[
\JJ\hspace{-2pt}=\hspace{-2pt}\left\{\J' \subseteq\hspace{-2pt} \mathcal{J}\colon\left|\J'\right| \hspace{-2pt}\geq\hspace{-2pt} 2, \left|\left\{j \in \J'\colon j_{k}\hspace{-2pt}=\hspace{-2pt}1\right\}\right| \hspace{-2pt}\leq\hspace{-2pt} 1, \forall k \in \R\right\}\hspace{-2pt},
\]
and
\begin{equation*}
\SmJp =\left\{k\in\R\colon \exists j\in\J', j_k=1\right\}\cup\Big(\Si\cap\big(\bigcap_{j\in\J'}\Sj\big)\Big).
\end{equation*}}
\end{remark}

\begin{remark}
    \label{rm:M_Construction}
\hg{Note that in the $\M$-Construction, for $\J\subseteq \Ki$, $|\J|\leq 1$, we have $\MJ=\varnothing$ because there are no $\J'\subseteq \J$ with $|\J'|\geq 2$ and hence, $\JJ = \varnothing$. This is consistent with our goal to form codewords of the minimum-weight: if $\J=\varnothing$ then $\bc=\bgi$ itself has weight $\wm$; if $\J=\{j\}$, then $\bc = \bgi\oplus \bgj$ also has weight $\wm$ due to the definition of $\Ki$.}
\end{remark}

Fig. \ref{fig:venn_diag} demonstrates the $\M$-construction, in particular, how to find $\mJp$ for every $\Jp$. 
\begin{figure}[ht] 
    \centering
    \includegraphics[width=0.9\columnwidth]{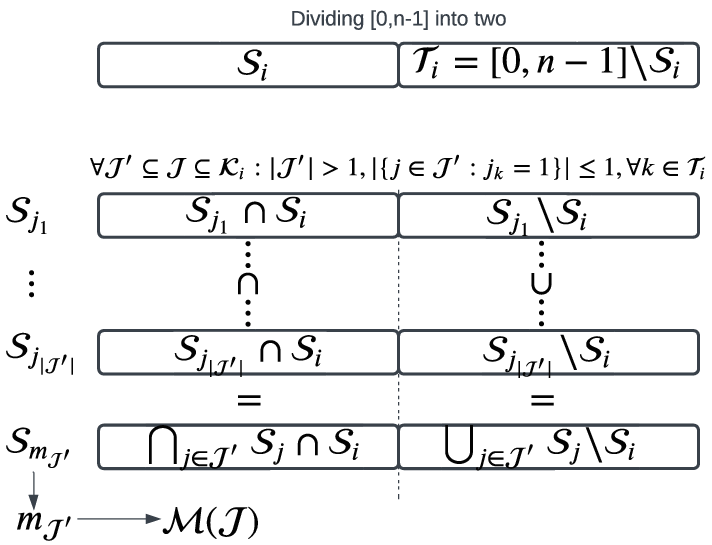}
    \caption{Illustration of how $\mJp$ is obtained for every $\Jp$.} 
    \label{fig:venn_diag}
\end{figure}

Fig. \ref{fig:venn_diag2} shows the Venn diagram associated with various sets defined in relation to the formation of minimum weight codewords of polar codes. 
\begin{figure}[ht] 
    \centering
    \includegraphics[width=0.9\columnwidth]{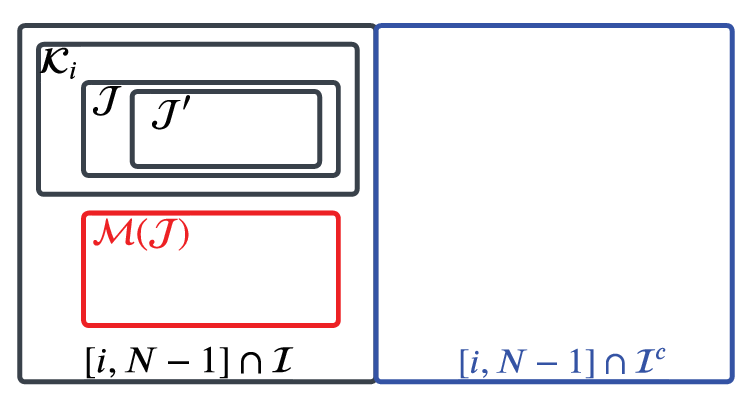}
    \caption{Venn diagram of the sets defined for indices in $[i,N-1]$.} 
    \label{fig:venn_diag2}
\end{figure}

We provide below a few examples to demonstrate the $\M$-Construction.

\begin{example}\label{ex:MJ1}
\hg{Let $n = 4$, $N = 2^n=16$, and $i=3=(0011)_2$. Then, $\Si=\{0,1\}$, $\Ti = [0,3]\setminus \Si=\{2,3\}$, and 
\[
\begin{split}
\Ki &= \{5,6,7,9,10,11\}\\
&= \left\{(0101)_2,(0110)_2,(0111)_2,(1001)_2,(1010)_2,(1011)_2\right\}\hspace{-2pt}.
\end{split}
\]
Take
\[
\begin{split}
\J &= \{5,6,7,9,10\}\\
&= \left\{(0101)_2,(0110)_2,(0111)_2,(1001)_2,(1010)_2\right\}\subset \Ki.
\end{split}
\]
We have $\S_5=\{0,2\}$, 
$\S_6=\{1,2\}$, 
$\S_7=\{0,1,2\}$, 
$\S_9=\{0,3\}$, 
$\S_{10}=\{1,3\}$.
Therefore, $\S_5\setminus \S_3 = \{2\}$, 
$\S_6\setminus \S_3 = \{2\}$,
$\S_7\setminus \S_3 = \{2\}$,
$\S_9\setminus \S_3 = \{3\}$,
$\S_{10}\setminus \S_3 = \{3\}$. As $\JJ$ consists of the subsets $\J'\subseteq \J$, $|\J'|\geq 2$, that satisfy that $\Sj\setminus \Si$, $j \in \J'$, are all disjoint, we have
\[
\JJ = \{\{5,9\}, \{5,10\}, \{6,9\}, \{6,10\},\{7,9\},\{7,10\}\}.
\]
From \eqref{eq:S_mJ}, we obtain $\SmJp$ for all $\J'\in \JJ$ as follows.
\[
\S_{m_{\{5,9\}}} = \big((\S_5\setminus \S_3)\cup (\S_9\setminus \S_3)\big) \cup \big(\S_3 \cap \S_5 \cap S_9\big) = \{0,2,3\}.
\]
\[
\S_{m_{\{5,10\}}} = \big((\S_5\setminus \S_3)\cup (\S_{10}\setminus \S_3)\big) \cup \big(\S_3 \cap \S_5 \cap S_{10}\big) = \{2,3\}.
\]
\[
\S_{m_{\{6,9\}}} = \big((\S_6\setminus \S_3)\cup (\S_9\setminus \S_3)\big) \cup \big(\S_3 \cap \S_6 \cap S_9\big) = \{2,3\}.
\]
\[
\S_{m_{\{6,10\}}} = \big((\S_6\setminus \S_3)\cup (\S_{10}\setminus \S_3)\big) \cup \big(\S_3 \cap \S_6 \cap S_{10}\big) = \{1,2,3\}.
\]
\[
\S_{m_{\{7,9\}}} = \big((\S_7\setminus \S_3)\cup (\S_9\setminus \S_3)\big) \cup \big(\S_3 \cap \S_7 \cap S_9\big) = \{0,2,3\}.
\]
\[
\S_{m_{\{7,10\}}} = \big((\S_7\setminus \S_3)\cup (\S_{10}\setminus \S_3)\big) \cup \big(\S_3 \cap \S_7 \cap S_{10}\big) = \{1,2,3\}.
\]
These supports correspond to $\mJp = 13, 12, 12, 14, 13, 14$. Therefore, according to \eqref{eq:set_J_h}, $|\mathfrak{J}_h(\J)|=2$ for $h=12, 13, 14$. As the cardinalities of $\mathfrak{J}_h(\J)$ are even for all $h=12,13,14$, according to \eqref{eq:set_M}, $\MJ=\varnothing$.}
\end{example}

\begin{example}
\label{ex:MJ2}
\hg{
We assume the same parameters $n=4$ and $i=3$ as in Example~\ref{ex:MJ1} but pick
 $\J= \{5,6,9,10\}$. Then $\JJ = \{\{5,9\},\{5,10\},\{6,9\},\{6,10\}\}$. As already computed in Example~\ref{ex:MJ1}, we have $
\S_{m_{\{5,9\}}} = \{0,2,3\}$, 
$\S_{m_{\{5,10\}}} = \{2,3\}$, $\S_{m_{\{6,9\}}} = \{2,3\}$, and $\S_{m_{\{6,10\}}} = \{1,2,3\}$.
These sets correspond to $\mJp = 13, 12, 12, 14$. Therefore, according to \eqref{eq:set_J_h}, $|\mathfrak{J}_h(\J)|=1$ (odd) for $h = 13, 14$ and $|\mathfrak{J}_h(\J)|=2$ (even) for $h = 12$. By \eqref{eq:set_M}, $\MJ=\{13, 14\}$.}
\end{example}

As a corollary of Theorem~\ref{thm:decomposition}, we can provide a lower bound on the number of minimum-weight codewords of a code $\CI$ (including RM and polar codes). It was established earlier in~\cite{bardet}, by analyzing the permutation group of polar codes, that this bound is the exact number of minimum-weight codewords. We provide a MATLAB script in Appendix~\ref{app:matlab} that computes this number (i.e., the error coefficient).
We discuss in detail the implicit connection between our work and the work in~\cite{bardet} at the end of this section. 

\begin{corollary} \label{trm:2}
If $\I\subseteq [0,N-1]$ satisfies the Partial Order Property then for every $i \in \B$ where $\B=\{i \in \I\colon \w(\bgi) = \wm\}$, the number of minimum-weight codewords of the code $\CI$ lying in the coset $\CiI$ satisfies
\begin{equation}\label{eq:bound_Admin_coset}
\Aiwm \geq 2^{|\mathcal{K}_i|},
\end{equation}
where $\Ki$ is given in Definition~\ref{def:Ki}.  
As a consequence, 
\begin{equation}\label{eq:bound_Admin}
\Awm = \sum_{i \in \B} \Aiwm
\geq \sum_{i \in \B} 2^{|\Ki|}.
\end{equation}
\end{corollary}
\begin{proof}
From Theorem~\ref{thm:decomposition}, we know that for every $i\in\B$ and $\J\subseteq\Ki$, there exists a set $\M(\J)\subseteq([i+1,N-1]\cap\I)\setminus\Ki$ such that \eqref{eq:decomposition} holds. 
In other words, any combination of row $i$ and rows in a subset $\J\subseteq\Ki$ gives a $\wm$-weight codeword. As $\Ki$ has $2^{|\Ki}|$ subsets, the $\M$-Construction provides $2^{|\Ki}|$ distinct minimum-weight codewords for $\CI$. Thus, $\Aiwm\geq 2^{|\Ki|}$ as claimed.
\end{proof}

We observe that the upper bound on the number of minimum-weight codewords proved by Bardet \textit{et al.}~\cite{bardet_arxiv} doesn't require that $\I$ must satisfy the Partial Order Property (referred to as decreasing monomial codes in their work). We restate the upper bound part of their result (see the proof of~\cite[Proposition 12]{bardet_arxiv}) using our terminology below.

\begin{proposition}[\cite{bardet_arxiv}]
\label{pro:upper_bound}
For an arbitrary set $\I\subseteq [0,N-1]$, let $\wm$ be the minimum weight of $\CI$, and $i\in I$ such that $\w(\bgi)=\wm$. Then the number of minimum-weight codewords in $\CiI\subseteq \CI$ (see Definition~\ref{def:CI}) satisfies $\Aiwm \leq 2^{|\Ki|}$.
\end{proposition}

\begin{figure*}
    \centering
    \setlength\fboxsep{5mm}
    \shadowbox{
        \parbox{0.92\textwidth}{ 
        Every row $\bgi,i\in\I$ of the transform matrix $\bGN$, where $\w(\bgi)=\wm$, can form a minimum-weight codeword in combination with the rows in every subset $\mathcal{J}\subseteq \mathcal{K}_i$ and the corresponding set $\M(\J) \subseteq (\I\cap[i+1,N-1])\setminus \Ki$ as
        \begin{equation*}
        \w\big(\bgi\oplus\underbrace{\bigoplus_{j\in\J}\bgj}_{\textnormal{core rows}} \oplus\underbrace{ \bigoplus_{m\in\M(\J)}\bgm}_{\textnormal{balancing rows}}\big) =  \wm.
        \end{equation*}
        The \textit{core} rows in the formation of minimum-weight codewords belong to a subset $\J$ of the set $\Ki$ defined as follows.
        \[
            \Ki \triangleq \{j \in [i+1,N-1]\colon \w(\bgj)\geq \w(\bgi+\bgj)=\w(\bgi)\}.
        \]
        The rows in the set $\M(\J)$ are called \textit{balancing} rows as their inclusion brings the weight of the sum down to $\wm$ if needed. The set $\MJ$ can be constructed by the $\M$-Construction described in Section~\ref{sec:decomposition}.
        The codewords formed by the leading row $\bgi$ belong to the coset $\CiI$ defined in Definition \ref{def:CI}.
        Note that $\bgi$ itself is a minimum-weight codeword. The information set $\I\subseteq [0,N-1]$ is assumed to satisfy the Partial Order Property. The construction above can be extended to all rows in $[0,N-1]$ (see \eqref{eq:wt_eq_Ci_gf}, Section~\ref{sec:PAC_Admin}).
            
        Since every subset of $\Ki$ corresponds to a different minimum-weight codeword, the total number of such codewords in every coset $\CiI$ equals the total number of subsets of $\Ki$, that is $2^{|\Ki|}$. Hence, the total number of minimum-weight codewords, matching the result in \cite{bardet}, is
        $$\Awm = \sum_{{i \in \I\colon \w(\bgi) = \wm}} 2^{|\Ki|}.$$
        }
    }
    \caption{A summary of the construction of minimum-weight codewords in polar codes.}
    \label{fig:summary_decomp}
\end{figure*}

\subsection{The Connection to the Permutation-Group-Based Approach by Bardet \textit{et al.}~\cite{bardet, bardet_arxiv}} 

Bardet \textit{et al.}~\cite{bardet, bardet_arxiv} use the transpose of $\bG_2$ instead in their constructions of RM/polar codes. Each row in $\bG_N$ indexed by $i\in [0,N-1=2^n-1]$ corresponds to the monomial $g_i\triangleq x_0^{i_0}x_1^{i_1}\cdots x_{n-1}^{i_{n-1}}\in \ff_2[x_0,\ldots,x_{n-1}]/(x_0^2-x_0,\ldots,x_{n-1}^2-x_{n-1})$, where $(i_0i_1\cdots i_{n-1})$ is the binary representation of $i$. The row $\bgi$ of $\bG_N$ is obtained by evaluating the monomial $g_i$ at the binary representations of all the column indices $c \in [0,N-1]$. They also define a partial order $\preceq$ on the monomials, which is equivalent to our partial order given in Definition~\ref{def:PO}. A set of monomials $\I$ (corresponding to our index set $\I\subseteq [0,N-1]$) is called \textit{decreasing} if and only if ($f\in \I$ and $g\preceq f$) implies that $g \in \I$. 
They show that the permutation group of the code $\CI$, which is generated by $\{\bgi\colon i \in \I\}$, contains LTA$(n,2)$, which consists of the transformations of the form $\bx \mapsto \bA\bx+\bb$, where $\bA = (a_{k,h})\in \ff_2^{n\times n}$ is a \textit{lower-triangular} matrix over $\ff_2$ with $a_{k,k}=1$ for all $0\leq k \leq n-1$, and $\bb=(b_0,\ldots,b_{n-1})\in \ff_2^n$ (see \cite[Theorem 2]{bardet_arxiv}). More specifically, under the transformation $(\bA,\bb)$, a monomial $g=x_{k_1}\cdots x_{k_s}$ (for some $0<k_1<k_2<\cdots<k_s\leq n-1$) is mapped into $y_{k_1}\cdots y_{k_s}$, where $y_k = x_k + \sum_{h=0}^{k-1}a_{k,h}x_h + b_k$. 

It is shown in~\cite[Theorem 2, Proposition 12]{bardet_arxiv} that \textit{all} minimum-weight codewords of $\CI$ can be generated by the codewords in the orbits $\O(g_i)$ under LTA$(n,2)$ of the monomials corresponding to the rows $i\in \I$ that have maximum degree $r_+$ (corresponding to the indices $i\in \I$ satisfying $\w(\bgi)=\wm$ in our work). 
Moreover, to count the number of minimum-weight codewords, for such $i$, they demonstrate in~\cite[Propositions 8 and 9]{bardet_arxiv} that $|\O(g_i)|$ is equal to the number of different transformations $(\bA,\bb)$ where $b_k=0$ if $k \notin \texttt{idx}(g_i)$ and $a_{k,h}=0$ if $k\notin \texttt{idx}(g_i)$ or $h \in \texttt{idx}(g_i)$, here $\texttt{idx}(g_i)\triangleq \{k_1,\ldots,k_s\}$. Based on Young diagrams, such $|\O(g_i)|$ can be determined explicitly based on the binary representation of $i$ (Bardet \textit{et al.}~\cite[Propositions 10 and 11]{bardet_arxiv}). It turns out that $|\O(g_i)|$ is exactly the same as our $2^{|\Ki|}$ (see Corollary~\ref{trm:2}). The reason is that the number of free entries (taking either $0$ or $1$) in a valid $(\bA,\bb)$ as described above is precisely equal to $|\Ki|$ (not hard to verify using Lemma~\ref{lma:Ki}(a)). We also give an explanation of how our $\M$-Construction can be extracted from their formulation below. 

For each $i\in [0,N-1]$ satisfying $\deg(g_i)=r_+$, let $g_i=x_{k_1}\cdots x_{k_s}$. The minimum-weight codewords in $\O(g_i)$ are of the form $(\bA,\bb)g_i$ for all valid $\bA$ and $\bb$ as described in the previous paragraph. Such a codeword corresponds to the polynomial $y_{k_1}\cdots y_{k_s}$, which can be written as
\begin{equation*}
\begin{split}
&\big(x_{k_1}+\sum_{h = 0, h\notin \texttt{idx}(g_i)}^{k_1-1}a_{k_1,h}x_h+b_{k_1}\big)
\cdots\\&\cdots \big(x_{k_s}+\sum_{h = 0, h\notin \texttt{idx}(g_i)}^{k_s-1}a_{k_s,h}x_h+b_{k_s}\big),
\end{split}
\end{equation*}

where $a_{k_1,h},\ldots,a_{k_s,h}$ for all relevant $h$, and $b_{k_1},\ldots,b_{k_s}$ can be either $0$ or $1$.
Due to the structure of $\Ki$ (see Lemma~\ref{lma:Ki}(a)), by inspecting the expansion of the product above more closely, we can recover the $\J$ and the $\M$ sets in our $\M$-Construction as follows. Note that our notation is complementary to theirs, and so an appropriate but straightforward transformation will be required for the sets to match exactly.
\begin{itemize}
    \item The term $x_{k_1}\cdots x_{k_s} = g_i$ corresponds to the row $\bgi$ in our construction (see Theorem~\ref{thm:decomposition}).
    \item The terms obtained after expanding the sums 
    \begin{equation}
    \label{eq:expansion}
    \begin{split}
    &\big(\sum_{h = 0, h\notin \texttt{idx}(g_i)}^{k_1-1}a_{k_1,h}x_h\big)x_{k_2}\cdots x_{k_s},\ldots\\& \ldots,x_{k_1}\cdots x_{k_{s-1}}\sum_{h = 0, h\notin \texttt{idx}(g_i)}^{k_s-1}a_{k_s,h}x_h
    \end{split}
    \end{equation}
    and $b_{k_1}x_{k_2}\cdots x_{k_s}$, $\ldots$, $x_{k_1}\cdots x_{k_{s-1}}b_{k_s}$ correspond to the rows $j \in J\subseteq \Ki$. More specifically, the terms including $\bA$-entries  correspond to $j\in \Ki$ with $|\Sj|=|\Si|$, where the terms including $\bb$-entries correspond to the $j\in \Ki$ with $|\Sj|=|\Si|+1$. Depending on whether the entries in $\bA$ and $\bb$ are 0 or 1, we have different subsets $J$ of $\Ki$. 
    \item The remaining terms in the product corresponds to the rows $m \in \M$ in our $\M$-Construction. 
\end{itemize}
From here, it can also be seen that the number of minimum-weight codewords from the orbit of $g_i$ is equal to two to the power of the number of free entries (can be assigned any value in $\ff_2$) in $\bA$ and $\bb$, which can be easily proven to be the same as $2^{|\Ki|}$. In the language of monomials and permutation groups~\cite{bardet, bardet_arxiv}, our code modification procedures in later sections perturb the set $\I$ by, e.g., removing a row/monomial $g_j$ that contributes to the formation of a large number of orbits, while adding back a row/monomial $g_i$ that has a small orbit, which is further reduced by half due to the removal of $j$. Note that $g_j$ contributes to the orbit formation of $g_i$ if it appears as a term in \eqref{eq:expansion}, e.g. $g_j=x_h x_{k_2}\cdots x_{k_s}$ for some valid index $h$, and hence, the removal of $j$ will effectively eliminate one free entry, e.g. $a_{k_1,h}$, making at least half of the minimum-weight codewords in the orbit of $g_i$ disappear.
On the other hand, as $I$ is decreasing, $g_i$ should not contribute to the orbit formation of any other $g_{i'}$, $i'\in \I$. The modification can be applied further to achieve extra reduction on $\Adm$. 
Our explicit formulation of the set $\Ki$ makes this modification process more transparent.

Before moving on, we summarize the construction of minimum-weight codewords and its application in the numeration of such codewords in Fig.~\ref{fig:summary_decomp}.


\subsection{Applications of the Minimum-Weight Codewords Characterization for Polar Codes} 
So far, we have explicitly characterized the row combinations involved in the formation of minimum-weight codewords and then used them to enumerate minimum-weight codewords as one of the potential applications. In the following sections, we shall see how this knowledge can help to improve the error coefficient of polar codes by a simple modification. This is not the only way to employ the minimum-weight codeword characterization in code design. For instance, instead of modifying polar codes, one can start from a low-order Reed-Muller code as a polar subcode and obtain a different polar-like code while considering the number of minimum-weight codewords. 

In rate-compatible polar coding, we use a pattern $\mathcal{P}$ (a certain set of bit indices) to shorten the codewords. One can easily explain and count the reduction in the number of minimum-weight codewords of shortened polar codes by considering the intersection of the shortening pattern, set $\mathcal{P}$, and set $\MJ$, that is, by checking $\mathcal{P}\cap\MJ\neq\emptyset$ for every $\J\subseteq\Ki$. This approach was used in \cite{gu23rate} to analyze the impact of shortening on the error coefficient of the PAC codes. 

When precoding is performed before polar coding, what we learned from the formation of minimum-weight codewords can be used to explain the impact of precoding on the weight distribution of a polar code, as will be discussed in Section \ref{sec:PAC_Admin}. This understanding was further used for a semi-closed-form enumeration of PAC codes in \cite{rowshan22fast} and for two different approaches to design a precoder for PAC codes in \cite{gu23improved,rowshan23min}. 

Hence, we can classify the applications of the main contribution of this work into three categories: 1) deterministic enumeration of polar codes and their variants, 2) design of polar-like codes and precoding, and 3) analysis of the impact of any code modification (such as shortening) or precoding on the weight distribution and error correction performance. 
The next section focuses on the application of minimum-weight codewords characterization in code design, followed by its application to explain the reduction of minimum-weight codewords in PAC coding in Section \ref{sec:PAC_Admin}.

\section{Error Coefficient-improved Codes}\label{sec:Ad_improv}
In this section, leveraging what we know about the structure of minimum-weight codewords of a polar code in Section~\ref{sec:decomposition}, we propose a procedure to construct new codes with fewer minimum-weight codewords.

Consider a polar code $\CI$, where $\I$ is constructed by the conventional methods such as density evolution (DE) or those used to approximate the DE.
We define the set $\B$ (which was used in Corollary~\ref{trm:2} as well) and $\Bc$ as follows:
\begin{equation}
    \B\triangleq\{ i\in \I: \w(\mathbf{g}_i)= \wm\},
\end{equation}
\begin{equation}
    \Bc\triangleq\{ i\in \I^c: \w(\mathbf{g}_i)= \wm\}.
\end{equation}
For each $j\in \B$, let us also define the set $\Ej$ and $\Dj$ as follows. 
\begin{equation}\label{eq:Ej}
    \begin{split}
        \Ej & \triangleq\left\{ i\in[0,j-1]: j\in \Ki, |\Si| =|\Sj| \right\} \\ &=\{i \in \B\cup \Bc \colon j \in \Ki\},
    \end{split}
\end{equation}
and
\begin{equation}
   \Dj \triangleq \Ej\cap\B = \{ i\in\B: j\in\mathcal{K}_i \}.
\end{equation}

\begin{remark}
The set $\mathcal{E}_j$ is formed by the \emph{right-swap operation} on $\bin(j)$ which is the opposite of the operation performed on $\bin(j)$ to form set $\mathcal{K}_j$. 
\end{remark}

Our first idea to start from an existing code $\CI$ and generate a new code $\CIp$, where $\I'$ is obtained from $\I$ by removing an index $j$ while adding a new index $i\notin \I$. The key point is to select $j$ and $i$ so that $\bgj$ contributes to the formation of more minimum-weight codewords in $\CI$ than $\bgi$ does in $\CIp$. Note that $\bgj$ can contribute to a minimum-weight codeword as a coset leader, as a row in the $\J$-part, or as a row in the $\M$-part (see Theorem~\ref{thm:decomposition}). Additionally, we choose $i\in \Ej$, or equivalently, $j \in \Ki$, to further reduce the number of minimum-weight codewords emerging due to the addition of $i$: with the removal of $j$ from $\I$, $|\Ki\cap \I^\prime|\leq |\Ki\setminus \{j\}| = |\Ki|-1$. 

\begin{proposition}\label{prop:Admin_reduction}
Suppose that $\I\subseteq [0,N-1]$ satisfies the Partial Order Property.
Given $j\in\B$ 
and $i\in(\Ej\cap\Bc)$
satisfying 
\begin{equation}\label{eq:cond_prop}
\bigg(\sum_{x\in\mathcal{D}_j} 2^{|\mathcal{K}_x|-1}\bigg)+2^{|\mathcal{K}_j|}>2^{|\mathcal{K}_i|-1},
\end{equation}
then
\begin{equation}\label{eq:A_dmin_reduction}
    \Adm(\I^\prime) \leq \Adm(\I) - \bigg(\bigg(\sum_{x\in\mathcal{D}_j} 2^{|\mathcal{K}_x|-1}\bigg)+2^{|\mathcal{K}_j|}-2^{|\mathcal{K}_i|-1}\bigg),
\end{equation} 
where $\I^\prime$ is $\I^\prime=\{i\}\cup \big(\I\setminus \{j\}\big)$.
\end{proposition}
\begin{proof}
First, after removing $j$ from $\I$ to obtain a new set of indices of the information bits $\I^{\prime\prime}=\I\setminus \{j\}$, the number of minimum-weight codewords satisfies the following inequality.
\begin{equation}\label{eq:new_A_dmin1}
    \Adm(\I^{\prime\prime})\leq\sum_{x\in\mathcal{D}_j} 2^{|\mathcal{K}_x|-1} + \sum_{y\in\B\setminus(\mathcal{D}_j\cup \{j\})} 2^{|\mathcal{K}_y|}.
\end{equation}
Hence, the number of minimum-weight codewords is reduced by at least $\big(\sum_{x\in\mathcal{D}_j} 2^{|\mathcal{K}_x|-1}\big)+2^{|\mathcal{K}_j|}$, which is the left-hand side of \eqref{eq:cond_prop}. The first term of the sum, $\sum_{x\in\mathcal{D}_j} 2^{|\mathcal{K}_x|-1}$, reflects the reduction of $\Adm$ due to the contribution of $\bgj$ (as the $\J$-part) to the cosets $\C_x(\I)$ for $x\in \Dj$, while the second term, $2^{|\mathcal{K}_j|}$, is the contribution of $\bgj$ (as the coset leader) to $\C_j(\I)$. 
The equality holds in \eqref{eq:new_A_dmin1} when $\bgj$ does not contribute as the $\M$-part to any cosets. 

On the other hand, we claim that by adding $i$ to the set $\I^{\prime\prime}$, the total number of minimum weight codewords increases by at most $2^{|\mathcal{K}_i|-1}$, which is the right-hand side of inequality \eqref{eq:cond_prop}. 
Indeed, we note that as $i\notin\I$, the row $\bgi$ will only contribute to the formation of minimum-weight codewords of $\C(\I^\prime)$ as a coset leader because all $i'\in \I$ already have their sets $\K_{i'}$ (hence their $\J$-parts) in $\I$ and their $\M$-parts in $\I$ as well. 
Here, we are applying Proposition~\ref{pro:upper_bound} on $i$ and the set $\I^\prime=\I^{\prime\prime}\cup \{i\}$, noting that this proposition doesn't require the set to satisfy the Partial Order Property. Furthermore, since $j\in\mathcal{K}_i$ (because $i \in \Ej$), and $j$ has been removed from $\I$, the row $\bgi$ contributes to at most $2^{|\mathcal{K}_i|-1}$ minimum-weight codewords in $\C(\I^\prime)$. 

Thus, as more minimum-weight codewords are lost than gained when going from $\CI$ to $\C(\I^\prime)$ due to \eqref{eq:cond_prop}, the inequality \eqref{eq:A_dmin_reduction} follows.
\end{proof}

According to Proposition~\ref{prop:Admin_reduction}, we can modify $\I$ to improve $\Adm$ given that there exists $i\in\I^c$ such that $\w(\mathbf{g}_i)=\wm$ and \eqref{eq:cond_prop} holds.  
It is clear that such a modification is impossible for Reed-Muller codes because $\I$ already contains all $i\in [0,N-1]$ with $\w(\bgi)=\wm$. 
For polar codes, to reduce the error coefficient, as a fast rule (which could be sub-optimal), one can look for $j\in\B$ with a large $|\mathcal{D}_j|=|\Ej\cap\B|$ and $i\in(\mathcal{E}_j\cap\Bc)$ with the smallest $|\mathcal{K}_i|$ that satisfies \eqref{eq:cond_prop}.

\begin{example}\label{ex:Admin(64,32,8)}
Let us take the polar code of $(64,32,8)$ where 
\begin{equation}
    \B=\{26,28,38,41,42,44,49,50,52,56\},
\end{equation}
\begin{equation}
    \Bc=\{7,11,13,14,19,22,25,35,37\}.
\end{equation}
Then with $j=56$, $\Dj=\B\setminus \{38\}$ has the largest size $|\Dj|=9$.
Observe that $\bin(38)=(100110)_2$ is not the result of the right-swap operation on $\bin(56)=(111000)_2$. On the other hand, we have $i=25$ where $i\in\Bc$ and $i\in\mathcal{E}_{56}$. Furthermore, $|\mathcal{K}_{25}|=8$ for $\bin(25)=(011001)_2$,  which is the smallest among $|\mathcal{K}_{i'}|$ for every $i'\in\Bc$ (actually, the alternative choice is 37). That is, by adding $i=25$ to set $\I$, the total number of codewords of minimum weight increases by $2^8$,  assuming no other changes in $\I$. Now, if we remove $j=56$ from $\I$, not only all $2^{|\mathcal{K}_j|}=2^{3}$ minimum-weight codewords from the coset $\C_j(\I)$ disappear, but also the number of minimum-weight codewords in every coset $\C_x(\I)$ for $x\in\mathcal{D}_{56}$ (and $x=i=25$) is reduced by half. 
Therefore, the reduction in the number of minimum-weight codewords going from $\I$ to $\I\cup\{i\}\setminus \{j\}$ is at least 
\begin{equation}
\begin{split}
&\big(\sum_{x\in\mathcal{D}_j} 2^{|\mathcal{K}_x|-1}\big)+2^{|\mathcal{K}_j|}-2^{|\Ki|-1} \\&=(64+32+64+32+16+32+16+8) + 8 - 128\\& = 264 + 8 - 128 = 144,
\end{split}
\end{equation}
as stated in Proposition~\ref{prop:Admin_reduction}. 
However, it turns out that we have achieved a larger reduction by this modification, which is $192 = 664-472>144$ (see Table~\ref{tab:modification}). The difference $48=192-144$ is due to the further loss of minimum-weight codewords in the coset led by the row 38. More specifically, by removing $j=56$ from $\I$, the number of minimum-weight codewords generated by this coset also reduces from 128 (as $|\mathcal{K}_{38}|=7$) to 80 because $j$ is a row in the $\M$-part corresponding to several $\J\subseteq \K_{38}$. This extra reduction is reflected (by the $\leq$ sign in \eqref{eq:A_dmin_reduction}) but not quantified in the statement of Proposition~\ref{prop:Admin_reduction}. 

\begin{table}
\centering
\begin{tabular}{ | c || c | c | }   
\hline
 $i$ & $A_{i,8}(\I)$ & $A_{i,8}(\I')$\\
\hline
  \cellcolor{green!50}25 &   & 128\\ 
  \hline
  26 & 128  & 64\\ 
  \hline
  28 & 64  & 32\\ 
  \hline
  38 & 128  & 80\\ 
  \hline
  41 & 128  & 64\\ 
  \hline
  42 & 64  & 32\\ 
  \hline
  44 & 32  & 16\\ 
  \hline
  49 & 64  & 32\\ 
  \hline
  50 & 32  & 16\\ 
  \hline
  52 & 16  & 8\\ 
  \hline
  \cellcolor{red!50}56 & 8   & \\ 
  \hline
   \textbf{Total} & \textbf{664}  & \textbf{472}\\ 
  \hline
\end{tabular}
\caption{The number of minimum-weight codewords of $\CI$ and $\CIp$ where $\I'=\I\cup\{25\}\setminus \{56\}$. The numbers of minimum-weight codewords in the cosets $\CiI$ and $\C_i(\I')$ are given in each row for $i\in \{25,26,\ldots,56\}$. By removing $j=56$ and adding $i=25$ to $\I$, 320 minimum-weight codewords are removed while 128 are added, resulting in a reduction of 192.}
\label{tab:modification}
\end{table}
\end{example}

\begin{remark}
Given $\I'=\{i\}\cup \big(\I\setminus \{j\}\big)$, the contribution of row $i$ where $\w(\bgi)=\wm$ and  $i\not\in\E_j$ is  $A_{i,d_{\min}}(\I')\leq A_{i,d_{\min}}(\I)=2^{|\Ki|}$ because $j$ could be in the set $\M$ associated with some $\J$ in the coset $\Ci(\I)$. For example, $A_{38,8}(\I')=80<128$ in Example \ref{ex:Admin(64,32,8)}.
\end{remark}
\begin{example}
In Example \ref{ex:Admin(64,32,8)}, $38\not\in\E_{56}$, as a result $A_{38,w_{min}}=80<2^{|\K_{38}|}$ where $2^{|\K_{38}|}=2^7=128$. Take $\J=\{42,52\}$ for example, observe that $m=56$ for this $\J$ but since $56\not\in\I'$, then $\w(\bg_{38}+\bg_{42}+\bg_{52})=12$, however,  $\w(\bg_{38}+\bg_{42}+\bg_{52}+\bg_{56})=8$. Note that $m\in\M$ for every $\J$ consists of $\{42,52\}$ and any subset of $\K_{38}\setminus\{42,52\}$, hence we expect to sabotage the formation of $2^{|K_{38}|-2}=2^5$ codewords at least.
\end{example}

\begin{corollary}\label{cor:gi_gr_wmin_no_incr}
Suppose that $\I \subseteq [0,N-1]$ satisfies the Partial Order Property. Pick an $i\in\I^c$ with $\w(\mathbf{g}_i)>\wm$ and set $\I^{\prime\prime}\triangleq\I\cup\{i\}$. Then $\Adm(\I)=\Adm(\I^{\prime\prime})$.
\end{corollary}
\begin{proof}
According to Corollary \ref{cor:geq_wi}, if $\w(\bgi)>\wm$, then $$\w(\bgi\oplus\bigoplus_{h\in\mathcal{H}}\mathbf{g}_h)> \wm,$$ where $\mathcal{H}\subseteq [i+1,2^{n}-1]$, therefore, no codewords with weight $\wm$ are introduced in the coset $\C_i$. Therefore,  $\Adm(\I^{\prime\prime})=\Adm(\I^{\prime})$, where $\I^\prime=\{i\}\cup \big(\I\setminus \{j\}\big)$.
\end{proof}


\section{Constructing New Codes: Procedure}
We can further reduce the error coefficient, $\Adm$, by repeating the process suggested in Proposition \ref{prop:Admin_reduction} for more pairs $(i,j)$. We propose a procedure\footnote{Python script available at https://github.com/mohammad-rowshan/Error-Coefficient-reduced-Polar-PAC-Codes}, detailed in Algorithm \ref{alg:code_design}, 
to find the pairs $(i,j)$ to modify the set $\I$ with the objective of reducing the number of minimum weight codewords.  That is, we are looking for $(i_1,j_1), \dots, (i_{\pi_{\max}},j_{\pi_{\max}})$ to modify the code $\C(\I)$ to obtain $\CIp$, where
\[
    \I^\prime = \{i_1,\dots, i_{\pi_{\max}}\} \cup (\I\setminus \{j_1,\dots, j_{\pi_{\max}}\}).
\]

This iterative procedure occurs in a loop in lines 4-38. 
As described before, the first step is to find the index $j_1\in\B$ that reduces the error coefficient $\Adm$ the most, or in mathematical notation, 
\begin{equation}
    j_1 = \argmax_{x\in\B} |\E_x\cap\B|.
\end{equation}
Recall that the reduction in the number of minimum-weight codewords due to the removal of $j_1$ (see Proposition~\ref{prop:Admin_reduction}) is at least 
\begin{equation}\label{eq:minus}
    \texttt{minus} \triangleq \bigg(\sum_{x\in\E_{j_1}\cap\B} 2^{|\mathcal{K}_x|-1}\bigg)+2^{|\mathcal{K}_{j_1}|}.
\end{equation}
One may want to find $j_1\in \B$ that maximizes the $\texttt{minus}$, which could be expensive. An alternative way is to simply look for the largest $|\D_j|=|\E_x\cap\B|$, as discussed in Section \ref{sec:Ad_improv} and in the proof of Proposition \ref{prop:Admin_reduction}. Note that both approaches may still be sub-optimal because as shown in Example~\ref{ex:Admin(64,32,8)}, $\texttt{minus}$ does not fully capture the possible reduction in the number of minimum-weight codewords when removing $j$. 
In lines 6-8, $j$ and $|\D_j|$ are collected in $\D$ and in line 11, the index $j$ corresponding to the largest $|\D_j|$ is obtained. Then the calculation of \texttt{minus} according to \eqref{eq:minus} is implemented in lines 12-14. 
Observe that for such $j_1$, we have $2^{|\mathcal{K}_{j_1}|}=\min_{x\in\B} 2^{|\mathcal{K}_x|}$ or $j_1=\max(\B)$. Further details and the application of this property are the subject of Section \ref{ssec:simpl_proced}. 
After finding such $j$, we remove it from the set $\I$. We denote the new set as the set $\I'=\I\setminus\{j\}$ (lines 33-35).  
The next step is to find a row $i$ that contributes the least to the error coefficient $\Adm$. The contribution of $i$ depends on whether it belongs to the set $\E_j\cap\Bc$ or $\Bc\setminus\E_j$ as follows:
\begin{equation}\label{eq:plus_cases}
\texttt{plus}=
    \begin{dcases*}
        2^{|\mathcal{K}_i|-1} & \text{ if } $i\in\E_j\cap\Bc$,\\
        2^{|\mathcal{K}_i|} & \text{ if } $i\in\Bc\setminus\E_j$.
    \end{dcases*}
\end{equation}
Lines 19-23 and 24-29 implement the two cases in \eqref{eq:plus_cases}, respectively. In subsequent iterations $\pi:1 \rightarrow \pi_{\max}$, we subtract $\pi$ from the exponents of $\texttt{minus}$ and $\texttt{plus}$ to account for the removal of $j$'s in lines 14 and 23. 
Note that since we removed $j$ from the set $\I$ in the first stage, i.e., $j\not\in\I'$, and on the other hand, we have $i\in\E_j\cap\Bc$, as a result, $j\not\in(\K_i\cap\I)$. Thus, the contribution of $i$ will be reduced to $2^{|\mathcal{K}_i|-1}$. This is the reason for choosing $i$ from $\E_j\cap\Bc$. We can check whether there exists some $i'\in\Bc\setminus\E_j$ such that $2^{|\mathcal{K}_{i'}|}<2^{|\mathcal{K}_i|-1}$. However, this is generally not the case. 

We can repeat this procedure a limited number of times, up to a suitable $\pi_{\max}$. However, the following iterations do not exactly follow Proposition \ref{prop:Admin_reduction} because the set $\I'$ no longer satisfies partial order property. Needless to mention that in each iteration, the sets $\B$ and $\Bc$ are changing due to updating the set $\I'$. This results in a difference in the set $\K_x\cap\I'$ for identical $x$ in every iteration (making $\K_x\cap\I'$ smaller due to the removal of larger indices from $\I'$). Note that the reduction estimated by removing $i$ is the lower bound. The reason is that the removed $j$ from the set $\I$ could be in the set $\M(\J)$ of some $\J\subseteq\K_i$. Therefore, as Theorem~\ref{thm:decomposition} suggests, in the absence of such $\M$, some of the codewords with minimum-weight in the coset $\C_i(\I')$ cannot be generated. This results in an additional reduction in the minimum weight codewords introduced by the coset $\C_i$, smaller than what is expected and consequently in a smaller error coefficient, $\Adm$. 

Also, if there exists some $i\in\I^c$ such that $\w(\bg_i)>\wm$, then this would have priority over an $i$ with $\w(\bg_i)=\wm$ because according to Corollary \ref{cor:gi_gr_wmin_no_incr}, adding this coordinate to $\I'$ will not contribute to $\Adm$. This is implemented in lines 15-17.

\begin{algorithm}
\caption{Code Modification}
\label{alg:code_design} 
\linespread{1.0}\selectfont 
\DontPrintSemicolon
\SetKwInOut{Input}{input}
\SetKwInOut{Output}{output}
\SetKwRepeat{Repeat}{do}{while} 
\SetKwFunction{func}{Subroutine}
\Input{Set of non-frozen indices $\I$, $\pi_{\max}$}
\Output{$\I$ (modified)}
    $\B \gets$ \text{extract all $i\in\I$ where $\w(\bgi)=\wm$}\;
    $\Bc \gets$ \text{extract all $i\in\I^c$ where $\w(\bgi)=\wm$}\;
    $\Bs \gets$ \text{extract all $i\in\I^c$ where $\w(\bgi)>\wm$}\;
    \For{$\pi$ \textnormal{in} $[1:\pi_{\max}]$}{
        $\D \gets \varnothing$, $\K \gets \varnothing$\;
        \For{$x$ \textnormal{in} $\B$}{
            \If{$|\E_x\cap\B|>0$}{ 
                $\D\gets \D\cup \{(x,|\E_x\cap\B|)\}$ \tcp*{Cf. \eqref{eq:Ej}}
            }
        }
        \If{$D=\varnothing$}{
            \text{Break}\;
        }
        $j \gets $ \text{Find $x$ associated with the largest} $|\E_x\cap\B|$ in $\D$\;
        $\texttt{minus} \gets 2^{|\mathcal{K}_j|-(\pi-1)}$ \tcp*{Cf. Lemma \ref{lma:Ki}.a} 
        \For{$x$ \textnormal{in} $\E_j\cap\B$}{
            $\texttt{minus} \gets \texttt{minus}+2^{|\mathcal{K}_x|-\pi}$\;
        }
        \uIf{$|\Bs|>0$}{
            $i\gets \max(\Bs)$\;
            \texttt{plus} $\gets$ 0, \texttt{paired} $\gets$ True\;
        }\Else{
            \If{$|\E_j\cap\Bc|>0$}{ 
                \For{$x$ \textnormal{in} $\E_j\cap\Bc$}{
                    $\K\gets  \K \cup \{(x,|\K_x|)\}$\;
                }
                $i \gets $ \text{Find $x<\min(\B)$ with the smallest} $|\K_x|$ in $\K$\;
                $\texttt{plus} \gets 2^{|\mathcal{K}_i|-\pi}$\;
            }\ElseIf{$|\Bc|>0$}{ 
                \For{$x$ \textnormal{in} $\Bc$}{
                    $\K\gets \K \cup \{(x,|\K_x|)\}$\;
                }
                $i' \gets $ \text{Find $x$ associated with the smallest} $|\K_x|$ in $\K$\;
                \If{\textnormal{\texttt{plus }}$>2^{|\mathcal{K}_{i'}|}$}{
                    \texttt{plus} $\gets$ $2^{|\mathcal{K}_{i'}|}$, $i \gets i'$\;
                }
            }
            \If{$\textnormal{\texttt{plus}} < \textnormal{\texttt{minus}}$}{
                \text{Remove $i$ from $\Bc$}\;
                \texttt{paired} $\gets$ True\;
            }
        }
        \uIf{\textnormal{\texttt{paired}} $=$ \textnormal{True}}{
            \text{Remove $j$ from $\B$}\;
           $\I \gets(\I\cup \{i\} \setminus \{j\})$
        }\Else{
            \text{Break}\;
        }
    }
\end{algorithm}

It is worth mentioning that the resulting information set $\I'$ using this procedure or the simplified procedure in the next section does not necessarily satisfy the partial order property. 
Algorithm \ref{alg:code_design} illustrates the procedure discussed above. In this procedure, the iterations are limited to $\pi_{\max}$. Additionally, in lines 30-32, we have a stopping criterion of $\texttt{plus}<\texttt{minus}$. This could be useful when $\pi_{\max}$ is considered large. In Section \ref{sec:Pe_vs_Admin}, we will discuss the need to balance reliability and error coefficient. The parameter $\pi_{\max}$ is chosen at the turning point where further improvement of the error coefficient does not improve the error correction performance of the code but degrades it.

\subsection{Simplified Procedure for Code Design}\label{ssec:simpl_proced}
The procedure introduced above requires operations on the sets and finding the largest or smallest elements in the sets. Some of these operations can be replaced with simpler operations based on prior knowledge of the polar transform and partial ordering. Here, we review Algorithm~\ref{alg:code_design} and find equivalent operations that are simpler. The procedure in general can be divided into two operations: 1) remove the indices that contribute the most to the error coefficient in the set $\I$ and 2) add the indices that contribute the least to the error coefficient in the set $\I$. These two operations are performed interactively by index pairs (up to $\pi_{\max}$ pairs) in Algorithm~\ref{alg:code_design}. In the following, we find equivalents for these two operations. 

We start by finding the indices that contribute the most to the error coefficient. 
In the first iteration of this algorithm, finding an $x$ that gives the largest $|\E_x\cap\B|$ in $\D$ (line 11 of Algorithm~\ref{alg:code_design}) is straightforward. By definition, the set $\B$ includes every $i\in \I$ with $\w(\mathbf{g}_i)= \wm$, then $\E_x$ intersects most of the elements in $\B$ when $\bin(x)$ has the form $\{1\}^q+\{0\}^r$ where $q=\log_2 \wm$ and $r=n-q$. In this notation, $\{1\}^q$ denotes a string of 1's repeated $q$ times, and $+$ is used for concatenation. This $x$ is $x \succeq i$ for every $i\in\B\setminus\{x\}$. That is, the rest of the elements in $\B$ can be obtained by single or multiple right-swap operations on $\bin(x)$. Hence, the $x$ that gives the largest $|\E_x\cap\B|$ in $\D$ is basically the element in $\B$ that has the largest index. The other candidates to be removed from $\I$ in the following iterations can be approximated by choosing the second and third largest indices in $\B$. Our observation shows that in the case of $\pi_{\max}=3$, we get identical index candidates to remove from $\I$. Therefore, the $\pi_{\max}$ largest indices in the set $\B$ are chosen. 

For the second operation, that is, selecting the least contributing indices of the set $\I^c$ that will play the role of coset leader, we choose a different approach. Assuming that the reduction (minus) in the error coefficient is significantly larger than the addition (plus), that is,
$$\texttt{minus}\gg\texttt{plus},$$
then instead of finding the least contributing indices, we can simply choose the most reliable bit-channels in set $\Bc$ to be added to the set $\I'$ if $\Bs$ is empty. Obviously, if $|\Bs|>0$, we prioritize adding the elements of the set $\Bs$. 

This simplified approach can be summarized in Algorithm \ref{alg:simplified_code_design}. 
Suppose that we have a reliability-ordered sequence $\Q_0^{N-1}$ such that $W_N^{(\Q_0)}\leq W_N^{(\Q_1)}\leq\cdots\leq W_N^{(\Q_{N-1})}$. This sequence can be obtained from any method discussed in the Introduction as long as the partial ordering property is maintained. Then, we select $K+\pi_{\max}$ most reliable indices from the sequence $\Q_0^{N-1}$, that is,is, from $\Q_{N-K-\pi_{\max}-1}$ to $\Q_{N-1}$ given the minimum distance $d_{\min}$ of $\Q_{N-K-\pi_{\max}-1}^{N-1}$ and $\Q_{N-K-1}^{N-1}$ is identical. If not, we can reduce $\pi_{\max}\in[1,3]$. The rest of the procedure follows the approach discussed above, as illustrated in the algorithm.

\begin{algorithm}
\caption{Simplified Code Design Procedure}
\label{alg:simplified_code_design} 
\linespread{1.0}\selectfont 
\DontPrintSemicolon
\SetKwInOut{Input}{input}
\SetKwInOut{Output}{output}
\SetKwRepeat{Repeat}{do}{while} 
\SetKwFunction{func}{Subroutine}
\Input{reliability ordered sequence $\Q_0^{N-1}$, $\pi_{\max}$}
\Output{$\I'$}
    $\Bs \gets$ find up to $\pi_{\max}$ elements $i$ in $\{\Q_{N-K-2}, \cdots, \Q_{1}, \Q_{0}\}$ where $\w(\bin(i))=2\log_2(w_{min})$\;
    
    $\I' \gets \{\Q_{N-K-\pi_{\max}-1+|\Bs|}, \cdots, \Q_{N-2}, \Q_{N-1}\}\cup\Bs$\;
    
    \For{$\pi$ \text{in} $[1:\pi_{\max}]$}{
        $j \gets \max \{j\in\I':\w(\bin(j))=w_{min}\}$\; 
        $\I' \gets \I'\setminus\{j\}$\;
    }
\end{algorithm}

Table \ref{tbl:two-algs} compares the output of Algorithms \ref{alg:code_design} and \ref{alg:simplified_code_design}. As can be seen, the selected non-frozen rows to be frozen are identical in both algorithms.

\begin{table}[ht]
\addtolength{\tabcolsep}{-0.3pt} 
\centering
\caption{Comparison of Algorithms \ref{alg:code_design} and \ref{alg:simplified_code_design} in terms of the indices added to and removed from set $\I$.}
\begin{tabular}{ |@{} c @{}|@{} c @{}|}   
\hline

\begin{tabular}{ c | c  }
&\\
N & Alg.\\
&\\
\hline
  \multirow{2}{*}{64} 
   & 1 \\ 
   & 2 \\
\hline
  \multirow{2}{*}{256} 
   & 1 \\ 
   & 2 \\
\hline
  \multirow{2}{*}{512} 
   & 1 \\ 
   & 2 \\
\end{tabular}
&
\begin{tabular}{ c | c |  c | c }
\multicolumn{4}{c}{Code Rate (R)}\\
 \hline
 \multicolumn{2}{c|}{$1/4$} &  \multicolumn{2}{c}{$3/4$}\\
 \hline
Removed & Added & Removed & Added \\
\hline
   60,58,57 & 30,29,{\color{red}27}  & 48,40 &  18,{\color{red}12}\\ 
   60,58,57 & {\color{red}39},30,29  & 48,40 & {\color{red}33},18\\ 
\hline
   248,244 & {\color{red}118},{\color{blue}63}  & 224,208,200 & 74,{\color{blue}23,15}\\ 
   248,244 & {\color{red}173},{\color{blue}63}  & 224,208,200 & 74,{\color{blue}23,15}\\ 
\hline
   496,488,484 & 335,315,311 & 448,416,400 & {\color{red}135},83,78\\ 
   496,488,484 & 335,315,311 & 448,416,400 & 83,78,{\color{red}58}\\ 
\end{tabular}\\
\hline
\end{tabular}
\label{tbl:two-algs}
\end{table}



As can be seen in Table \ref{tbl:two-algs}, for both algorithms, the largest index for each code in the `Removed' 
columns has the form of $\{1\}^q+\{0\}^r$ and the other indices are the result of the right-swap operation of the least significant bit (LSB) on the binary representation of the largest index. For example, $60=(111100)_2$ is the largest, and the second and the third are $58=(111010)_2$ and $57=(111001)_2$ where they all have a Hamming weight of 4. 
Furthermore, the indices added to set $\I$ in both algorithms are also similar (identical or different in one element) and according to our observation, the slight difference in some of the codes does not significantly change the error coefficient as illustrated in  Table \ref{tbl:two-algs-err-coeff}. Note that P+ and PAC+ denote polar codes and PAC codes, respectively, constructed by Algorithms \ref{alg:code_design} and \ref{alg:simplified_code_design}. 
As a result, the block error rate will remain almost the same. 
The indices highlighted in blue in Table \ref{tbl:two-algs} have a Hamming weight of $2\log_2(w_{min})$ and those shown in red highlight the differences between the results of the two algorithms. The codes with rate $1/2$ also follow this similarity in both algorithms; however, due to the limit of column width, we omitted them from the table.

\begin{table}[ht]
\addtolength{\tabcolsep}{-0.67pt} 
\centering
\caption{Comparison of Algorithms \ref{alg:code_design} and \ref{alg:simplified_code_design} in terms of resulting error coefficients.}
\begin{tabular}{ |@{} c @{}|@{} c @{}|}   
\hline

\begin{tabular}{ c | c  }
&\\
$N$ & Alg.\\
&\\
\hline
  \multirow{2}{*}{64} 
   & 1 \\ 
   & 2 \\
\hline
  \multirow{2}{*}{256} 
   & 1 \\ 
   & 2 \\
\hline
  \multirow{2}{*}{512} 
   & 1 \\ 
   & 2 \\
\end{tabular}
&
\begin{tabular}{ c | c |  c | c | c | c }
\multicolumn{6}{c}{Code Rate (R)}\\
 \hline
 \multicolumn{2}{c|}{$1/4$} &  \multicolumn{2}{c|}{$1/2$} &  \multicolumn{2}{c}{$3/4$}\\
 \hline
    P+ & PAC+ & P+ & PAC+ & P+ & PAC+ \\
\hline
   196 & 24  & 408 & 112 & 272 &  108\\ 
   220 & 44  & 408 & 184 & 304 &  216\\ 
\hline
   5912 & 568  & 77104 & 13904 & 272 &  216\\ 
   6424 & 608  & 77104 & 13904 & 272 &  216\\ 
\hline
   4048 & 748  & 18720 & 4412 & 13504 &  4832\\ 
   4048 & 748 & 18720 & 4412 & 13504 &  4832\\ 
\end{tabular}\\
\hline
\end{tabular}
\label{tbl:two-algs-err-coeff}
\end{table}
\section{Impact of Precoding on Error Coefficient}\label{sec:PAC_Admin} 
In this section, we consider convolutional precoding. The recently introduced polarization-adjusted convolutional (PAC) coding scheme can reduce the number of minimum-weight codewords. This reduction is a result of the inclusion of rows in $\I^c$ in the generation of codewords in the cosets \cite{rowshan-precoding}. Note that the convolutional precoding does not change the set $\B$. That is, the leaders of cosets $\C_i,i\in\B$ remain unchanged in the PAC coding. In this section, we study how precoding further reduces the number of minimum-weight codewords.

The input vector $\mathbf{u}=[u_0,\ldots,u_{N-1}]$ in PAC codes unlike polar codes is obtained by a convolutional transformation using the binary generator polynomial of degree $m$, with coefficients $\mathbf{p}=[p_0,\ldots,p_m]$ as follows: 
\begin{equation}\label{eq:precoding}
    u_i = \sum_{j=0}^m p_j v_{i-j},
\end{equation}

This convolutional transformation combines $m$ previous input bits stored in a shift register with the current input bit $v_i$ to calculate $u_i$. The parameter $m$ is known as the {\em memory} of the shift register, and by including the current input bit we have the {\em constraint length} $m+1$ of the convolutional code. 
Note that the convolutional precoding does not reduce the minimum distance of a polar code (and thus the minimum weight of non-zero codewords) due to Corollary \ref{cor:geq_wi}. This was shown in \cite[Lemma 1]{rowshan23min}. 

From a polar coding perspective, the vector $\mathbf{u}$ is equivalent to the vector $\mathbf{v}$ in the PAC coding by $\mathbf{p}=[1]$. 
To obtain similar combinations of rows in $\bGN$ to form a minimum-weight codeword, we need to have $u_a=1$ for every  $$a\in\{i\}\cup\J\cup\M,$$ hence we have 
\begin{equation}\label{eq:wt_gi_gk1_u}
w\big(\bgi\oplus\bigoplus_{j\in\mathcal{J}}\bgj \oplus \bigoplus_{m\in\MJ}\bgm\big) =  \wm,
\end{equation} 
where $\mathcal{J}\subseteq \mathcal{K}_i$ and $\M(\J) \subseteq \I\setminus \Ki$. Obviously, we need $u_b=0$ for any $$b\in\big(\I\cap[i,N-1]\big)\setminus\big(\{i\}\cup\J\cup\M\big).$$ Note that the values of elements in the vector $\mathbf{v}$ are not important as long as we get the desired vector $\mathbf{u}$ as a result of transmission in \eqref{eq:precoding}. That is, a different message vector $\mathbf{d}=\left[d_0, d_1, \ldots, d_{K-1}\right]$ (see Section II in \cite{rowshan-precoding} for more details) in the PAC coding may result in the same code as in the polar coding if they both have the same vector $\mathbf{u}$. If we represent the convolution operation in the form of a Toeplitz matrix, where the rows of a {\em convolutional generator matrix} $G$ are formed by shifting the vector $\mathbf{p} = (p_0,p_1,\ldots p_m)$ one element at a row, as shown in (\ref{eq:conv_gen}). 
\begin{equation}
\label{eq:conv_gen}
\mathbf{P}=\begin{bNiceMatrix}
p_0 & p_1 & \Cdots & p_m & 0 & \Cdots & & 0 \\
0 & p_0 & p_1 & \Cdots & p_m & & & \\
\Vdots & \Ddots & \Ddots & \Ddots & & \Ddots & \Ddots &\Vdots \\
 & & & & & & & 0 \\
 & & & & p_0 & p_1 & \Cdots & p_m \\
\Vdots & & & & \Ddots & \Ddots & \Ddots &\Vdots \\
 & & & & & & p_0 & p_1 \\
0 & \Cdots & & & & & 0 & p_0
\end{bNiceMatrix}
\end{equation}

Note that $p_0$ and $p_m$ by convention are always 1, hence it is an upper-triangular matrix. 
Then, we can obtain $\mathbf{u}$ by matrix multiplication as $\mathbf{u}=\mathbf{v}\mathbf{P}$. 
As a result of this pre-transformation, $u_f$ for $f\in\I^c$ may no longer be frozen (i.e., $u_f=0$) as in polar codes. Therefore, $u_f\in\{0,1\}$. 

The important point to note is that $v_f=0$ for $f\in\I^c$. To have codewords similar to those of the polar codes in \eqref{eq:wt_gi_gk1_u}, we need $u_a=1$ for $a\in\{i\}\cup\J\cup\M(\J)$. 

In general, depending on the convolution of the $m$ previous inputs, that is, $V=\sum_{k=1}^m p_k v_{x-k}$, we have $u_x=V+v_x$. Therefore, we can obtain $u_x=1$ by setting the current input $v_x$ as 
\begin{equation}\label{eq:v_x_based_on_V}
        v_x=
\begin{dcases*}
        1 & \text{if } $\sum_{k=1}^m p_k v_{x-k}=0$,\\
        0 & \text{otherwise},\\
\end{dcases*}
\end{equation}
and for $u_x=0$, we can do the opposite.

Hence, to get $u_x=1$ for $x\in\{i\}\cup\J\cup\M(\J)$ and $u_x=0$ for $x\in\big(\I\cap[i,N-1]\big)\setminus\big(\{i\}\cup\J\cup\M(\J)\big)$, we can set $v_x$ according to the general rules mentioned above. 
However, we have no control over the values of $u_f$ for $f\in\I^c$ as a result of \eqref{eq:precoding} knowing that $v_f=0$ by default.  Consequently, there will be another term as $\bigoplus_{f\in\F(\J)}\bgf$ for $\F(\J)\subseteq\I^c\cap[i,N-1]$ in \eqref{eq:wt_gi_gk1_u} which is inevitable. This term may increase the weight of the generated codeword:

\begin{equation}\label{eq:wt_gi_gk_gf}
w\big(\bgi\oplus\bigoplus_{j\in\mathcal{J}}\bgj \oplus \bigoplus_{f\in\F(\J)}\bgf \oplus \bigoplus_{m\in\mathcal{M}(\J\cup\F(\J))}\bgm  \big) \geq  \wm.
\end{equation}

Note that the equality in \eqref{eq:wt_gi_gk_gf} is our conjecture based on observations and requires a rigorous proof by extending the $\M$-construction as discussed in this section, and it is an open problem. 

\begin{example}\label{ex:Admin_PAC_pi_1_2}
Suppose that we have the polar code of (64,32,8). If we modify it as
$\I'=\big(\{25\}\cup\I\setminus\{56\}\big)$, the number of minimum-weight codewords broken into cosets will be as shown in the table below before precoding and after precoding.
\begin{table}[ht]\label{tbl:one-pair}
\begin{center}
\caption{Number of minimum-weight codewords in the coset $\Ci(\I\cap[i,N-1])$ and $\Ci(\I'\cap[i,N-1])$ for the code (64,32,8) and $\pi=1$, with and without precoding where $\mathbf{c}=[1,0,1,1,0,1,1]$.}
\begin{tabular}{ | c || c | c | c | c |}   
\hline
 & \multicolumn{2}{c|}{Polar} & \multicolumn{2}{c|}{PAC}\\
\hline
 $i$ & $A_{i,8}(\I)$ & $A_{i,8}(\I')$ & $A_{i,8}(\I)$ & $A_{i,8}(\I')$\\
\hline
  \cellcolor{green!50}25 &   & 128 &  & 0\\ 
  \hline
  26 & 128  & 64 & 0 & 0\\ 
  \hline
  28 & 64  & 32 & 0 & 0\\ 
  \hline
  38 & 128  & 80 & 128 & 64\\ 
  \hline
  41 & 128  & 64 & 128 & 64\\ 
  \hline
  42 & 64  & 32 & 64 & 32\\ 
  \hline
  44 & 32  & 16 & 32 & 16\\ 
  \hline
  49 & 64  & 32 & 64 & 32\\ 
  \hline
  50 & 32  & 16 & 32 & 16\\ 
  \hline
  52 & 16  & 8 & 16 & 8\\ 
  \hline
  \cellcolor{red!50}56 & 8   &  & 8 & \\ 
  \hline
   Total & 664  & 472 & 472 & 232\\ 
  \hline
\end{tabular}
\end{center}
\label{tab:enumerat_PAC2}
\end{table}

\begin{table}[ht]\label{tbl:two-pairs2}
\begin{center}
\caption{Number of minimum-weight codewords in the coset $\Ci(\I\cap[i,N-1])$ and $\Ci(\I'\cap[i,N-1])$ for the code (64,32,8) and $\pi=2$, with and without precoding where $\mathbf{c}=[1,0,1,1,0,1,1]$. }
\begin{tabular}{ | c || c | c | c | c |}   
\hline
 & \multicolumn{1}{c|}{Polar} & \multicolumn{3}{c|}{PAC}\\
\hline
 $i$ & $A_{i,8}(\I)$ & $A_{i,8}(\I)$ & $A_{i,8}(\I')$ & $A_{i,8}(\I')$\\
\hline
  \cellcolor{green!50}22 &  & & & 0\\ 
  \hline
  \cellcolor{green!50}25 & &  & 0 & 0\\ 
  \hline
  26 & 128 & 0 & 0 & 0\\ 
  \hline
  28 & 64 & 0 & 0 & 0\\ 
  \hline
  38 & 128 & 128 & 64 & 32\\ 
  \hline
  41 & 128 & 128 & 64 & 32\\ 
  \hline
  42 & 64 & 64 & 32 & 16\\ 
  \hline
  44 & 32 & 32 & 16 & 8\\ 
  \hline
  49 & 64 & 64 & 32 & 16\\ 
  \hline
  50 & 32 & 32 & 16 & 8\\ 
  \hline
  \cellcolor{red!50}52 & 16 & 16 & 8 & \\ 
  \hline
  \cellcolor{red!50}56 & 8 & 8 & & \\ 
  \hline
   Total & 664 & 472 & 232 & 112\\ 
  \hline
\end{tabular}
\end{center}
\label{tab:enumerat_PAC3}
\end{table}

\end{example}



For any set $\J\subseteq \K_i$ and its associated set $\M(\J) \subseteq \I\setminus \Ki$, if as a result of precoding in \eqref{eq:precoding} there exists set $\F(\J)\subseteq\I^c\cap[i,N-1]$ such that  
$|\Sf\setminus\Si|>1$ for at least one $f\in\F(\J)$, then we have

\begin{equation}\label{eq:wt_gr_Ci_gf_geq}
w\big(\bgi \oplus \bigoplus_{j\in\J}\bgj \oplus \bigoplus_{f\in\F(\J)}\bgf \oplus \bigoplus_{m\in\M(\J\cup\F(\J))}\bgm  \big) >  \wm.
\end{equation} 
Recall from Lemma \ref{lma:Ki} (properties of $\K_i$)  and Theorem~\ref{thm:decomposition} that in order to get a codeword with weight $\wm$, the members of $\Ki$, here $\K_i\cup\F(\J)$, should satisfy the condition $|\Sj\setminus\Si|=1$ for $j\in\K_i\cup\F(\J)$. The set $\M$ also is defined based on this property for every set $\J\subseteq\K_i\cup\F(\J)$. Now, if there exists an element $j\in\K_i\cup\F(\J)$ such that $|\Sj\setminus\Si|>1$, Theorem~\ref{thm:decomposition} is no longer valid. Hence, 

\begin{equation}\label{eq:wt_gr_Ci_gf_geq2}
w\big(\bgi \oplus \bigoplus_{j\in\J\cup\F(\J)}\bgj \oplus \bigoplus_{m\in\M(\J\cup\F(\J))}\bgm \big) >  \wm.
\end{equation} 

\begin{example}
In the coset $\C_{26}$ and $\C_{28}$ of the polar code (64,32,8) discussed in Example \ref{ex:Admin_PAC_pi_1_2}, there exists some $\F(\J)\subset\{32,33,34,35,36,37\}\cup\{40\}\cup\{48\}$ where $\bin(26)=(011010)_2$, and $\bin(28)=(011100)_2$. Observe that $|\S_{f}\setminus\S_{26}|>1$ for $f\in\{33,35,36,37\}$ and $|\S_{f}\setminus\S_{28}|>1$ for $f\in\{33,34,35,37\}$. Hence, as can be seen, $A_{d_{\min},i}$ for $i=\{26,28\}$ reduced to zero after precoding. This reduction occurs as a result of the inevitable combination of at least one $f$ with the aforementioned condition with $\J_1\subseteq\{26,28\}$ along with or without $\J_2\subseteq\Ki\setminus\J_1$. For this specific $\mathbf{p}$, there is always such an $f$ however, with shorter $\mathbf{p}$, i.e., smaller $m$, this may not be the case.
\end{example}

The opposite of \eqref{eq:wt_gr_Ci_gf_geq} is expected to be true if $|\Sf\setminus\Si|=1$ for every $f\in\F(\J)$, then we have

\begin{equation}\label{eq:wt_eq_Ci_gf}
w\big(\bgi \oplus \bigoplus_{j\in\J}\bgj \oplus \bigoplus_{f\in\F(\J)}\bgf \oplus \bigoplus_{m\in\M(\J\cup\F(\J))}\bgm  \big) =  \wm.
\end{equation} 
Fig. \ref{fig:venn_diag_pac} shows the Venn diagram associated with various sets defined so far in relation to the formation of minimum weight codewords after precoding. 
\begin{figure}[ht] 
    \centering
    \includegraphics[width=0.9\columnwidth]{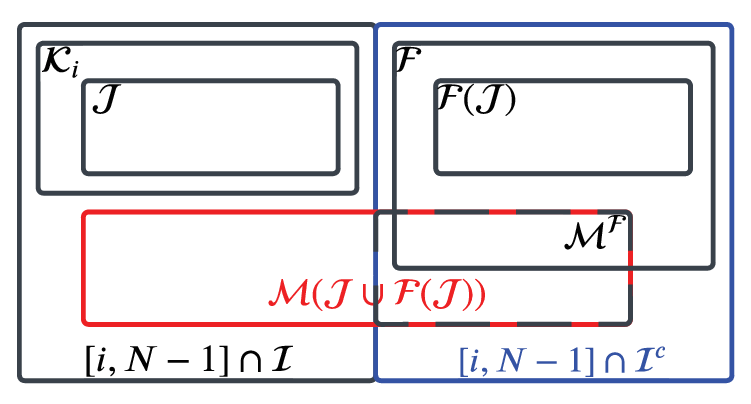}
    \caption{Venn diagram of the sets defined for indices in $[i,N-1]$ including the frozen indices as a result of precoding. Here, $\F$ is $\F\triangleq\{f\in\I^c\cap[i,N-1]:u_f=1\}$. } 
    \label{fig:venn_diag_pac}
\end{figure}


\begin{remark}
Observe that although every $f\in\F(\J)$ satisfies condition $|\S_f\setminus\Si|=1$ for inclusion in the set $\Ki$, it is linearly dependent on some $\J\subseteq\Ki$, so any $f\in\F(\J)$ satisfying the condition will not increase the number of row combinations that give codewords with weight $\wm$.
\end{remark}

Note that although the weight in \eqref{eq:wt_eq_Ci_gf} remains the same, the generated codewords are not the same as the combination without involving $\bigoplus_{f\in\F(\J)}\bgf$ of which we see in polar codes.

\begin{example}
In the coset $\C_{41}$, $\C_{42}$ and $\C_{44}$ of the polar code (64,32,8) discussed in Example \ref{ex:Admin_PAC_pi_1_2}, there exists $\F(\J)=\{48\}$ where $\bin(41)=(101001)_2$, $\bin(42)=(101010)_2$, $\bin(44)=(101100)_2$, and $\bin(48)=(110000)_2$. Observe that $|\S_{48}\setminus\S_{41}|=|\S_{48}\setminus\S_{42}|=|\S_{48}\setminus\S_{44}|=1$. Hence, as can be seen, $A_{d_{\min},i}$ for $i=\{41,42,44\}$ remains unchanged after precoding. This is the case for $\C_{38}$ where $\F(\J)=\{40,48\}$ the property $|\Sf\setminus\Si|=1$ follows.
\end{example}

\begin{remark}
Observe that if $\F(\J)=\I^c\cap[i,N-1]=\varnothing$ for the coset $\Ci$, then precoding has no impact on $A_{i,\wm}$ as there is no $f\in\F(\J)$ to follow \eqref{eq:wt_gr_Ci_gf_geq}. Therefore, it will be the same as \eqref{eq:wt_gi_gk1_u}.
\end{remark} 

\begin{example}
In the coset $\C_{49}$, $\C_{50}$, $\C_{52}$, and $\C_{56}$ of the polar code (64,32,8) discussed in Example \ref{ex:Admin_PAC_pi_1_2}, there exists no set $\F(\J)$. Therefore, as can be seen, $A_{d_{\min},i}$ for $i=\{49,50,52,56\}$ remains unchanged after precoding.
\end{example}

\section{Reliability vs Error Coefficient}\label{sec:Pe_vs_Admin} 

As discussed in Section \ref{sec:Ad_improv}, we freeze the bit-channel(s) with the highest contribution(s) into $\Adm$. These bit-channels have relatively higher reliability compared to those not in the set $\I$. Then, to keep the code rate constant, we have to unfreeze the bit-channels with the lowest contribution into $\Adm$. Obviously, these bit-channels have  relatively lower reliability. Recall that we denote the index set of the non-frozen bits by $\I'$. 
For the decision on the entire codeword to be correct,  all individual decisions must be correct. 
If we decode such a polar code formed by set $\I$ with successive cancellation (SC) decoder, the block error event $E$ is a union over $\I$ of the event that the first bit error occurs, denoted by $E_i \triangleq \{\hat{u}_i \neq u_i \mid \hat{u}_1^{i-1}=u_1^{i-1}\}$ where $E=\bigcup_{i\in\I} E_i$. Let $E^c$ denote the event that the block is decoded correctly, that is $\hat{u}_1^N=u_1^N$, then the probability of block error is obtained by \cite{mori}
\begin{equation}
    P_{SC}(\I) = P(E)=1-P(E^c)=1-\prod_{i \in \I}(1-P(E_i)),
\end{equation}
where $P(E_i)$ is the probability of error at bit-channel $i$ at a particular noise power or SNR assuming that bits 0 to $i-1$ are decoded successfully. Note that $P(E_i)=0$ for any $i\in\I^c$. As a result of modifying the set $\I$ by swapping high-reliability bit-channels with low-reliability ones, the error coefficient improves as we saw in the previous sections; however, we will have 
\begin{equation}
    P_{SC}(\I) \leq P_{near-ML}.
\end{equation}

If we use a near ML decoder, the block error rate depends more on $d_{\min}$ and $\Adm$ for linear codes as discussed in Appendix~\ref{ssec:BLER_Admin}. However, in the context of polar codes, if we continue to improve the error coefficient, $\Adm$, of the code, assuming that we can do it by the aforementioned process multiple times, we will be weakening the block code in terms of reliability. There is a turning point where further improvement of the error coefficient not only does not improve the block error rate further but it results in the error correcting performance degradation. That is, the gain due to the improvement of the error coefficient cannot overcome the degradation due to the loss of block reliability, $1-P_{SC}(\I) \geq 1-P_{SC}(\I')$. Fig. \ref{fig:BLER_Admin} illustrates the block error rate (BLER) versus the error coefficient. The turning point at $E_b/N_0=3.5$ is $\pi=5$ bit-pairs, while for $E_b/N_0=2.5$, it is $\pi=3$ bit-pairs. This shows the importance of the error coefficient at relatively high SNR regimes. Unfortunately, this turning point cannot be found analytically. As a general rule, we can agree that as long as the error coefficient increases $\pi_{\max}$, the resulting gain can overcome the loss due to the block reliability. Note that if we design a code solely based on error coefficient, as union bound implies, at very high SNR regimes where the noise is small, the error coefficient plays the dominant role, and the power gain appears; although, at low and medium SNR regimes, the performance may not be competitive. In this work, we aim to target the medium and low SNR regimes; hence, we consider both reliability and error coefficient in code design by limiting $\pi_{\max}$.

\begin{figure}[ht] 
    \centering
    \includegraphics[width=1\columnwidth]{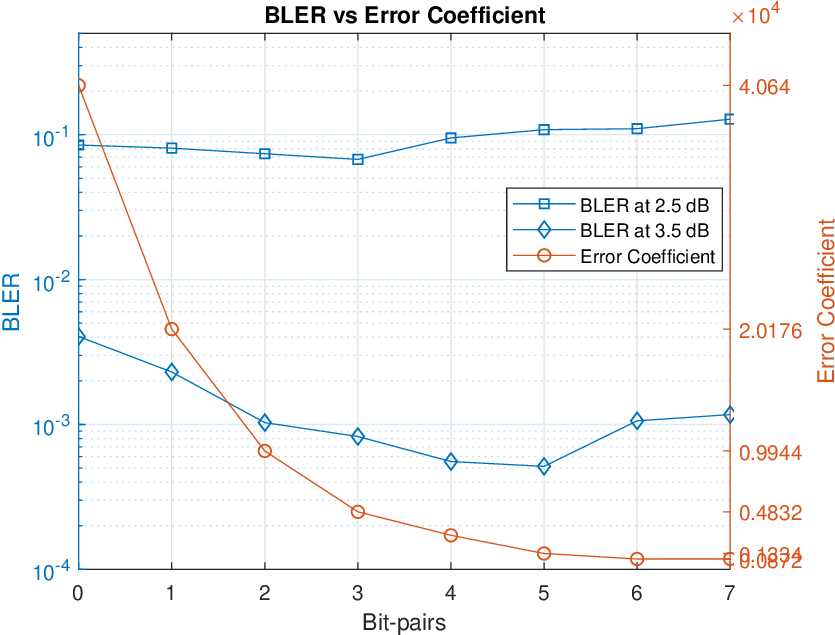}
    \caption{BLER at two $E_b/N_0$'s versus error coefficient $\Adm$ of PAC code of (512,384) under SC list decoding with L=16. The bit-pairs are the number of bit indices added to/removed from $\I$, denoted by $\pi_{\max}$ in the proposed procedure.} 
    \vspace{-5pt}
    \label{fig:BLER_Admin}
\end{figure}


\section{Numerical Results} \label{sec:results}
In this Section, we assess the amount of reduction in $\Adm$ based on the proposition and the method we proposed, and we observe the power gain resulting from the error correction performance of the improved codes. For this purpose, we choose three block-lengths ($N=64,256,512$), and three code rates ($R=1/4,1/2,3/4$) totaling nine different codes. We will construct new codes based on these nine codes and present the error coefficients and block error rates for the new codes with and without convolutional precoding. Therefore, we will compare 4x9=36 codes in total. 

Table \ref{tbl:Admin_for_P_PAC_plus} illustrates the error coefficient, $\Adm$, of the 36 aforementioned codes. Note that P+ and PAC+ represent the codes resulting from the rate profile $\I'$ corresponding to polar codes and PAC codes with the rate profile $\I$. As can be seen, the modified polar codes, P+, have a smaller error coefficient than PAC codes for all codes except for (256,64). This code and code (256,128) have special distributions of indices in $\I$ in the range $[0,N-1]$. In code (256,128), we have only two indices in $\I$ where $\w(g_i)=\wm$. Therefore, the new code has a larger $d_{\min}$.

To optimize the performance for BLER $10^{-2}-10^{-3}$, the density evolution method with Gaussian approximation (DEGA)  \cite{trifonov} is used to construct the base codes, that is, to obtain the set $\I$. Precoding in all PAC codes is performed with polynomial coefficients $\mathbf{p}=[1,0,1,1,0,1,1]$. 

\begin{table}[ht]
\addtolength{\tabcolsep}{-0.6pt} 
\centering
\caption{The minimum Hamming distance and the associated error coefficient, $\Adm$, of polar and PAC codes before and after modification. P+ and PAC+ represent the modified polar and PAC codes with rate profile $\I'$ while polar code (P) and PAC codes use the rate profile $\I$.}
\begin{tabular}{ |@{} c @{}|@{} c @{}|}   
\hline

\begin{tabular}{ c | c  }
&\\
N & Code\\
&\\
\hline
  \multirow{4}{*}{64} 
   & P \\ 
   & P+ \\
   & PAC \\ 
   & PAC+ \\ 
\hline
  \multirow{4}{*}{256} 
   & P \\ 
   & P+ \\
   & PAC \\ 
   & PAC+ \\ 
\hline
  \multirow{4}{*}{512} 
   & P \\ 
   & P+ \\
   & PAC \\ 
   & PAC+ \\ 
\end{tabular}
&
\begin{tabular}{ c | c | c | c | c | c }
\multicolumn{6}{c}{Code Rate (R)}\\
 \hline
 \multicolumn{2}{c|}{$1/4$} & \multicolumn{2}{c|}{$1/2$} & \multicolumn{2}{c}{$3/4$}\\
 \hline
 $d_{\min}$ & $\Adm$ & $d_{\min}$ & $\Adm$ & $d_{\min}$ & $\Adm$\\
\hline
   16 & 364 & 8 & 664 & 4 & 432\\ 
   16 & 196 & 8 & 408 & 4 & 304\\ 
   16 & 236 & 8 & 472 & 4 & 320\\ 
   16 & 24 & 8 & 112 & 4 & 108\\ 
\hline
   32 & 13336 & 8 & 96 & 8 & 82016\\
   32 & 5912 & 16 & 77104 & 8 & 28448\\ 
    32 & 2200 & 8 & 96 & 8 & 53456\\ 
    32 & 568 & 16 & 13904 & 8 & 6704\\ 
\hline
   32 & 13616 & 16 & 61024 & 8 & 49344\\ 
   32 & 4048 & 16 & 18720 & 8 & 13504\\ 
   32 & 6496 & 16 & 36256 & 8 & 40640\\ 
   32 & 748 & 16 & 4412 & 8 & 4832\\ 
\end{tabular}\\
\hline
\end{tabular}
\label{tbl:Admin_for_P_PAC_plus}
\end{table}

To evaluate the block error rate (BLER) of the codes, we transmit the modulated codewords using binary phase-shift keying (BPSK) scheme through additive white Gaussian noise (AWGN) channel. 
The lower bound for maximum likelihood (ML) decoding is obtained by counting the list decoding failures ($L=32$) when the Euclidean distance between the received signals $\mathbf{y}$ and the modulated transmitted signals $\mathbf{c}$ is greater than the Euclidean distance between the received signals $\mathbf{y}$ and the estimation of transmitted signals through decoding $\hat{\mathbf{c}}$, that is, when $||\mathbf{y}-\mathbf{c}||>||\mathbf{y}-\hat{\mathbf{c}}||$. Clearly, in this situation, ML decoding cannot be successful. 

As Fig. \ref{fig:FER64_1} to Fig. \ref{fig:FER512_3} show, the power gain of polar+ codes over polar codes and, in particular, PAC + codes over PAC codes is significant. Observe that polar+ codes are outperforming PAC codes for most of the codes evaluated in this section. When comparing PAC+ codes with CRC polar codes, we still observe that PAC+ codes outperform CRC-polar codes at the practical BLER range of $10^{-2}-10^{-3}$. For comparison purposes, we compare the BLERs with the lower bound for the BLER of finite block length codes called dispersion bound. The dispersion bounds with normal approximation \cite{polyanskiy} are obtained by the Laplace transform-based integration proposed in \cite{erseghe}.

Let us analyze the numerical results for each code rate separately. 
Fig. \ref{fig:FER64_1},\ref{fig:FER256_1},\ref{fig:FER512_1} illustrate block error rates for codes with block lengths $N=64,$ 256, and 512 and rate $R=1/4$. At this code rate, we observe a significant power gain of 0.2 to 0.5 dB for PAC+ over CRC-polar codes and PAC codes depending on the code length. The gain at short codes is larger. For  $N=64$, the BLER reaches the dispersion bound.  Note that as $N$ increases, appending a relatively short CRC does not cost in terms of code rate as much as it costs at short block lengths. Therefore, the gain relative to the CRC-polar reduces as $E_b/N_0$ increases. 

We can observe a similar performance gain for code rate $R=3/4$ in Fig. \ref{fig:FER64_3}, \ref{fig:FER256_3}, and \ref{fig:FER512_3}.

\begin{figure}[ht] 
    \centering
    \includegraphics[width=1\columnwidth]{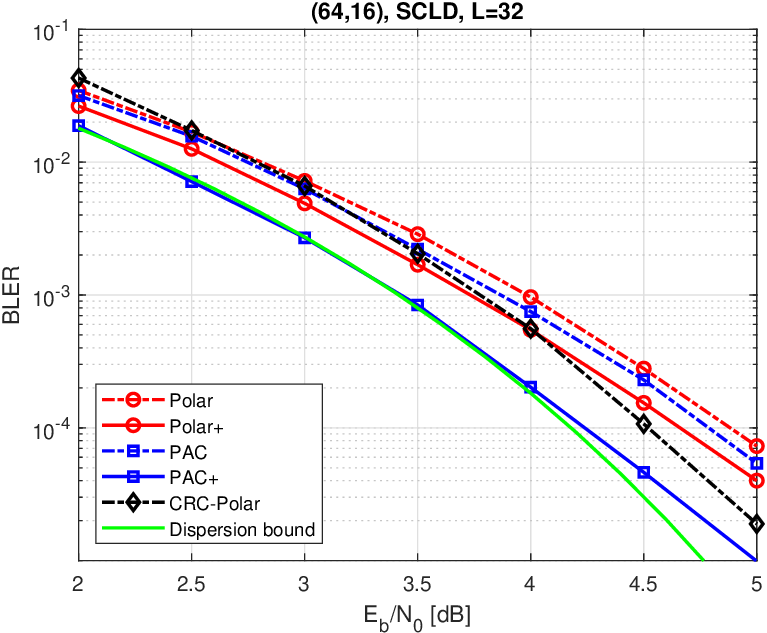}
    \caption{BLER Comparison of various (64,16)-codes. Parameters: design-SNR=4 dB, CRC: 0xA5, $\pi_{\max}=3$,  $\I'=\{30,29,27\}\cup\I\setminus\{60,58,57\}$.} 
    \vspace{-5pt}
    \label{fig:FER64_1}
\end{figure}
\begin{figure}[ht] 
    \centering
    \includegraphics[width=1\columnwidth]{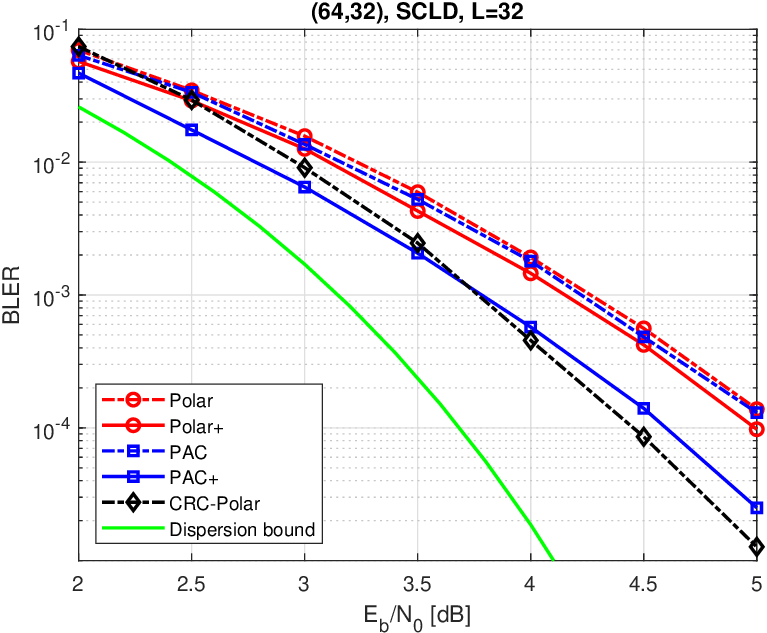}
    \caption{BLER Comparison of various (64,32)-codes. Parameters: design-SNR=4 dB, CRC: 0xA5, $\pi_{\max}=2$,  $\I'=\{25,22\}\cup\I\setminus\{56,52\}$.} 
    \vspace{-5pt}
    \label{fig:FER64_2}
\end{figure}
\begin{figure}[ht] 
    \centering
    \includegraphics[width=1\columnwidth]{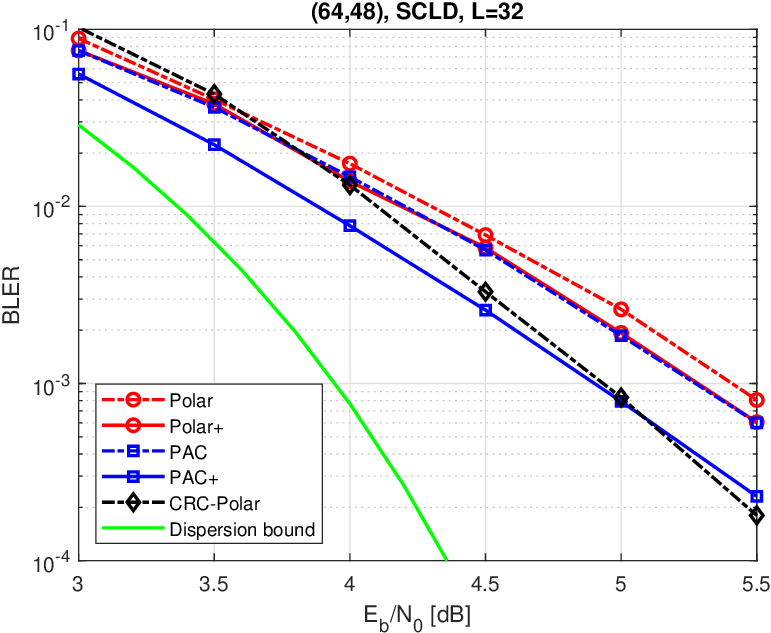}
    \caption{BLER Comparison of various (64,48)-codes. Parameters: design-SNR=2 dB, CRC: 0xA5, $\pi_{\max}=2$,  $\I'=\{22,18\}\cup\I\setminus\{48,40\}$.} 
    \vspace{-5pt}
    \label{fig:FER64_3}
\end{figure}

\begin{figure}[ht] 
    \centering
    \includegraphics[width=1\columnwidth]{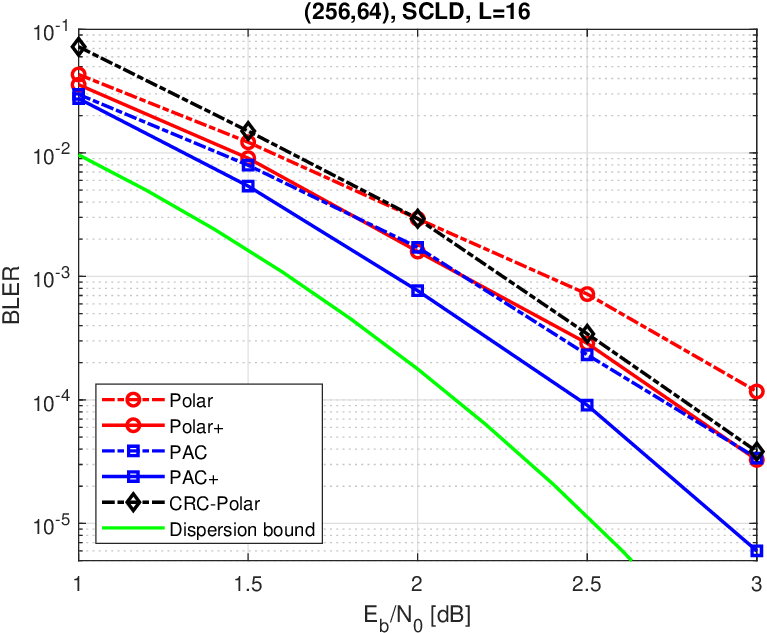}
    \caption{BLER Comparison of various (256,64)-codes. Parameters: design-SNR=4 dB, CRC: 0xA5, $\pi_{\max}=2$,  $\I'=\{118,63\}\cup\I\setminus\{248,244\}$.}
    \vspace{-5pt}
    \label{fig:FER256_1}
\end{figure}
\begin{figure}[ht] 
    \centering
    \includegraphics[width=1\columnwidth]{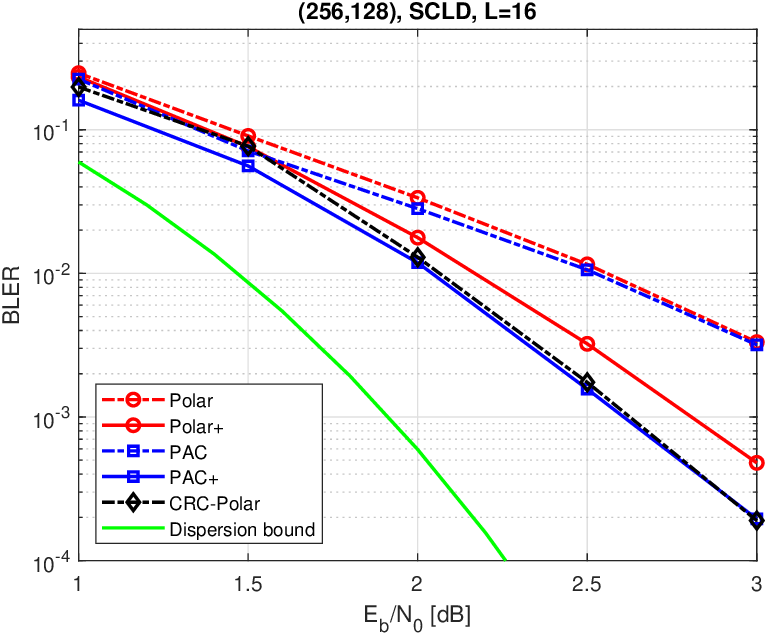}
    \caption{BLER Comparison of various (256,128)-codes. Parameters: design-SNR=2 dB, CRC: 0xA5, $\pi_{\max}=2$,  $\I'=\{149,147\}\cup\I\setminus\{224,208\}$.}
    \vspace{-5pt}
    \label{fig:FER256_2}
\end{figure}
\begin{figure}[ht] 
    \centering
    \includegraphics[width=1\columnwidth]{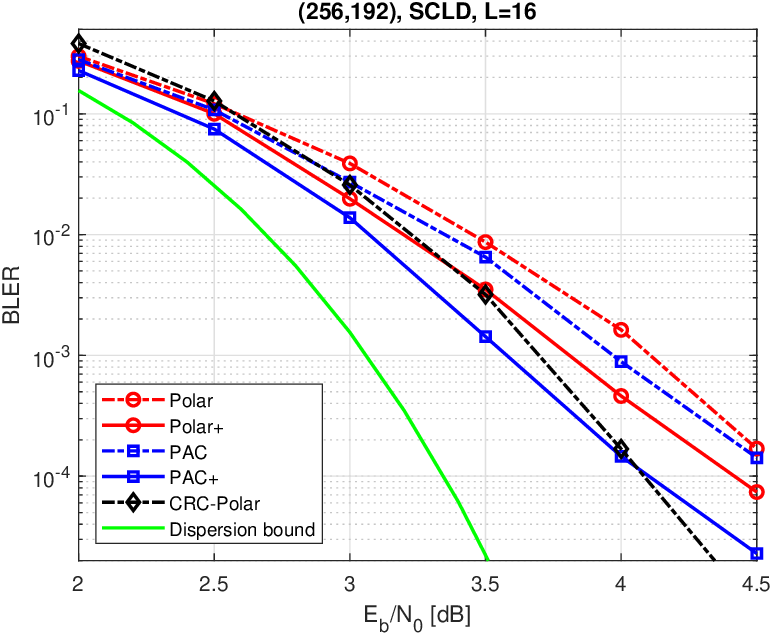}
    \caption{BLER Comparison of various (256,192)-codes. Parameters: design-SNR=4 dB, CRC: 0xA5, $\pi_{\max}=3$,  $\I'=\{74,23,15\}\cup\I\setminus\{224,208,200\}$.}
    \vspace{-5pt}
    \label{fig:FER256_3}
\end{figure}

\begin{figure}[ht] 
    \centering
    \includegraphics[width=1\columnwidth]{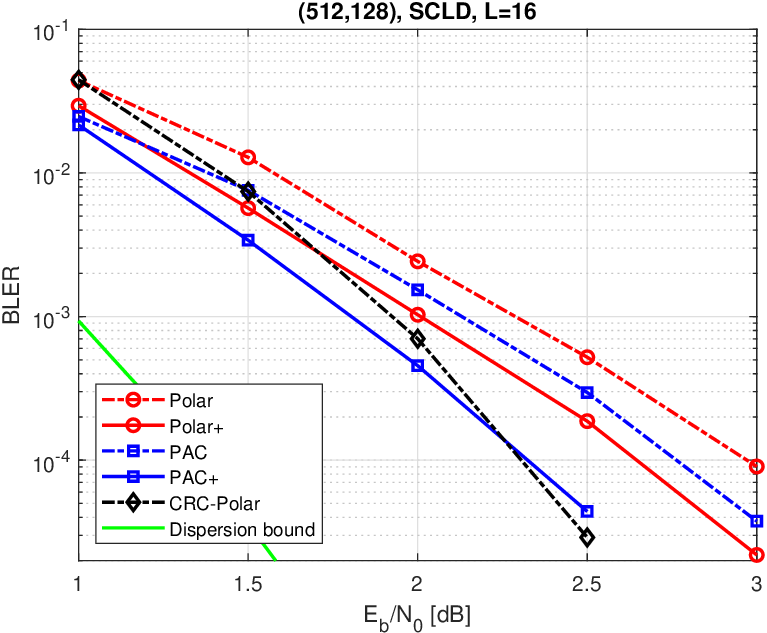}
    \caption{BLER Comparison of various (512,128)-codes. Parameters: design-SNR=2 dB, CRC: 0xC06, $\pi_{\max}=3$,  $\I'=\{335,315,311\}\cup\I\setminus\{496,488,484\}$.}
    \vspace{-5pt}
    \label{fig:FER512_1}
\end{figure}
\begin{figure}[ht] 
    \centering
    \includegraphics[width=1\columnwidth]{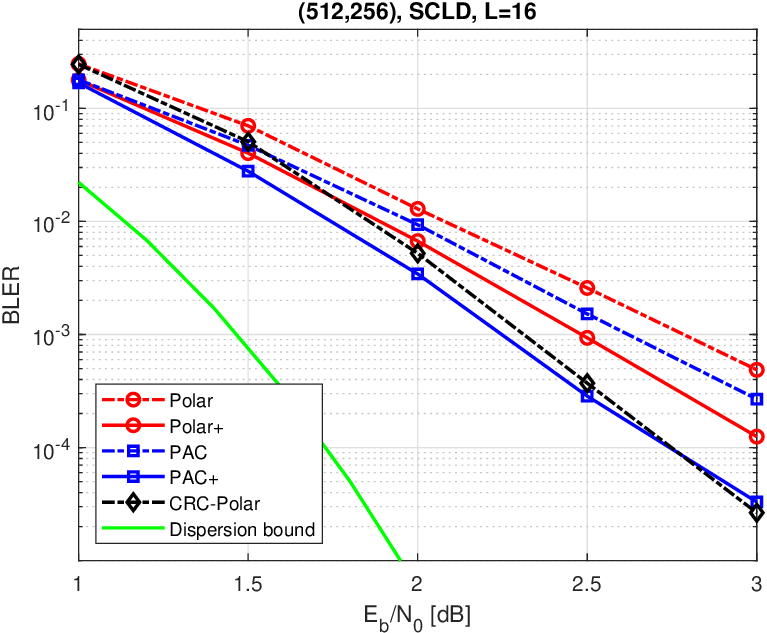}
    \caption{BLER Comparison of various (512,256)-codes. Parameters: design-SNR=2 dB, CRC: 0xC06, $\pi_{\max}=3$,  $\I'=\{283,279,271\}\cup\I\setminus\{480,464,456\}$.}
    \vspace{-5pt}
    \label{fig:FER512_2}
\end{figure}
\begin{figure}[ht] 
    \centering
    \includegraphics[width=1\columnwidth]{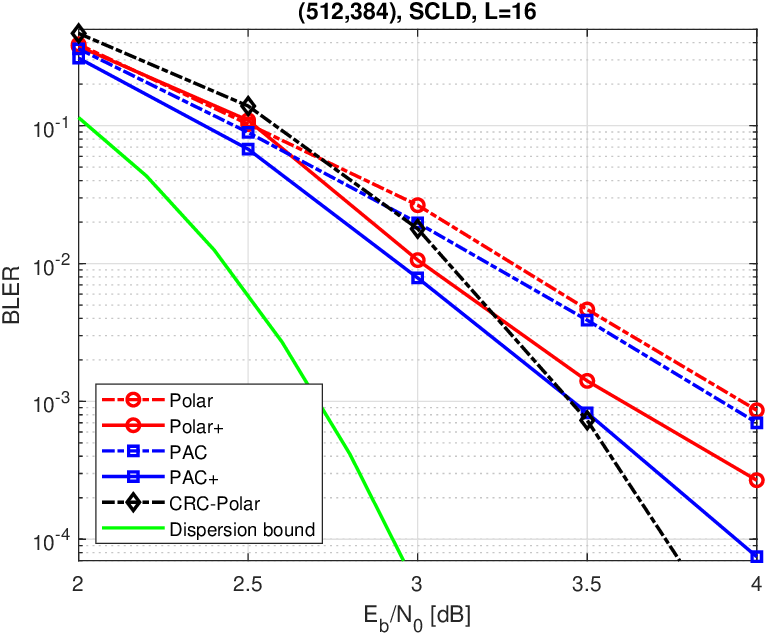}
    \caption{BLER Comparison of various (512,384)-codes. Parameters: design-SNR=4 dB, CRC: 0xC06, $\pi_{\max}=3$,  $\I'=\{135,83,78\}\cup\I\setminus\{448,416,400\}$.}
    \vspace{-5pt}
    \label{fig:FER512_3}
\end{figure}

For the code rate $R=1/2$ in Fig.  \ref{fig:FER64_2}, \ref{fig:FER256_2}, and \ref{fig:FER512_2}, it is observed that the power gain of PAC+ over CRC-polar at the practical BLER range of $10^{-2}-10^{-3}$ is smaller than at other rates. A careful observer would find a similar power gain of 0.4-0.6 dB for PAC+ codes over PAC codes as other code rates. Hence, CRC-polar might be performing better at this rate than at other rates. To explain this observation, let us divide the block errors that occur in the list decoding into two types: 1) Elimination error: when the correct sequence is eliminated before decoding the last bit, 2) Miss error: when the correct sequence remains in the list until decoding the last bit, but it is not the sequence with the highest likelihood. Clearly, the CRC mechanism as a genie that finds the correct sequence can reduce the miss error rate, but not the elimination error. On the other hand, we know that all information bits are mapped to high-reliability bit-channels at a low code rate while we will have a significant number of information bits transmitted through low-reliability bit-channels. Hence, at low code rates, both miss error rate and elimination error rate are relatively low because of exploiting high-reliability bit-channels hence, CRC will not have a significant impact at our target range. However, the elimination error rate is relatively high compared with the miss error rate at high code rates because the correct sequences may not survive due to the overall low reliability of exploited bit-channels. In this case, CRC cannot also be helpful. In the medium code rates, the CRC mechanism seems to be more effective, as the overall reliability of non-frozen bit-channel is at a moderate level. 

Note that the minimum distance of the new construction for the code (256,125), as Table \ref{tbl:Admin_for_P_PAC_plus} indicates, increased from 8 to 16. This is the main reason for the higher power gain of polar+ and PAC+ codes  as shown in Fig. \ref{fig:FER256_2} compared to other codes.

\section{Conclusion}
In this paper, we discover the combinatorial properties of polar transform $\mathbf{G}_N$ based on the row and column indices and then characterize explicitly all the row combinations involved in the formation of the minimum-weight codewords. In other words, we explicitly provide the decomposition of minimum-weight codewords into the rows of polar transform. The decomposed rows are classified into core and balancing rows. 
First, this characterization gives an elementary enumeration of the minimum-weight codewords based on core rows. Unlike other methods, it is based on explicitly counting all the row combinations resulting in minimum-weight codewords. The core application of this characterization is to explain how the error coefficient is reduced after convolutional precoding. Furthermore, we propose an exact and approximate method to significantly reduce the error coefficient of polar and PAC codes. Evaluation of the BLER of various codes shows that the designed codes can outperform CRC-polar codes and PAC codes in the practical BLER regime of $10^{-2}-10^{-3}$. The decomposition of minimum-weight codewords gives a significant insight into more analytical and practical works related to polar code modifications. Finally, in this work, we only considered the core rows in the code construction, taking the balancing rows into consideration seems to be a promising future direction. 



\appendix
\subsection{MATLAB Script for Enumeration}
\label{app:matlab}

The function \emph{err\textunderscore coeff} in MATLAB\textsuperscript{\texttrademark} language in the following listing can be used to obtain $d_{\min}$ and $\Adm$ given the inputs; the index set of the non-frozen bits, $\I$, and the code length $N$. Note that the index of non-frozen bits starts from $0$ and consequently the largest index is $N-1$. In line 2, the function finds the minimum $|\supp(i)|$ or sum of 1's in $\bin(i)$ for $i\in\I$. Then, we can get the minimum Hamming distance of the code, that is 
$d_{\min}=2^{|\supp(i)|}$ for every $i$ that gives the minimum $|\supp(i)|$. We collect all such $i$ in set $\B$ in line 4. Then, in the outer loop, we find $|\Ki|$ for every $i\in\B$. According to \eqref{eq:Ki_size}, the size of $|Ki|$ is the sum of $|\Ti|=\log_2(N)-|\supp(i)|$, in line 6, and $\sum_{k\in \Si}\sum_{\ell>k} \bar{i}_\ell$, in lines 7-10 or inner loop. Recall that $\Adm=\sum_{i\in\B} 2^{|\Ki|}$. This is being accumulated in line 11 in the outer loop.

\begin{lstlisting}
function [dmin, A_dmin] = err_coeff(I,N)
    d = min(sum(dec2bin(I)-'0',2));
    dmin = 2^d; n = log2(N); A_dmin = 0;
    B = find(sum(dec2bin(I)-'0',2)==d);
    for i = B'
        Ki_size = n - d;
        for x = find(dec2bin(I(i),n)-'0'==1)
            ii = dec2bin(bitxor(N-1,I(i)),n)-'0';
            Ki_size = Ki_size + sum(ii(1:x-1));
        end
        A_dmin = A_dmin + 2^Ki_size;
    end
end
\end{lstlisting}

Moreover, you may find a Python script for Algorithm \ref{alg:code_design} on https://github.com/mohammad-rowshan/Error-Coefficient-reduced-Polar-PAC-Codes. 

\subsection{Block Error Probability and the Number of minimum-weight Codewords of Polar Codes}\label{ssec:BLER_Admin}

The Hamming distance between two non-identical codewords $\mathbf{v}, \mathbf{w}$ in $\mathcal{C}$ is defined as $d(\mathbf{c},\mathbf{c}')=\w(\mathbf{c}+\mathbf{c}')$. It is known that the linear block codes can correct up to $\lfloor (\dist(\C)-1)/2\rfloor$ errors, where $\dist(\C)$ is the minimum Hamming distance of code $\dist(\C)$.  
In the linear block codes, $\mathbf{c}+\mathbf{c}'$ in $\ft$ gives another codeword in $\cC$, let us call it $\mathbf{c}''$, then
\begin{equation}\label{eq:d_min_w_min}
    \dist(\C)=\min\{\w(\mathbf{c}''),\mathbf{c}''\in\cC,\mathbf{c}''\neq \mathbf{0}\}=\wm.
\end{equation}
The minimum Hamming weight, in this paper we use its short form as minimum-weight, defines the error correction capability of a code.  
Besides minimum Hamming weight, the number of minimum-weight codewords is also important. 

It was shown in \cite[Sect. 10.1]{lin_costello} that for a binary input additive white Gaussian noise (BI-AWGN) channel at high $E_b/N_0$, the upper bound for block error probability of linear codes under soft-decision maximum likelihood (ML) decoding can be approximated by 
\begin{equation*}\label{eq:union_bound}
    P_e^{ML} \approx \Awm Q(\sqrt{2\dist(\C)\cdot R \cdot E_b/N_0}),
\end{equation*} 
where $\Awm$ denotes the number of minimum-weight codewords, a.k.a error coefficient, $Q(\cdot)$ is the tail probability of the normal distribution $\mathcal{N}(0,1)$, and $R$ is the code rate. As $\Awm$ is directly proportional with the upper bound for the error correction performance of a code, it can be used as a measure to anticipate the direction of change in the block error rate when $\Awm$ changes or in general it is a measure for relative performance of the codes under the same decoding. 

\subsection{Polar Transform and its Properties}
\label{app:polar_transform}

\hg{In this appendix, for self-completeness, we develop a few useful properties regarding the polar transform, some of which were known under equivalent formulations in the literature.} 

The polar transform matrix $\bGN$ is defined as the $n$-th Kronecker power of
\begin{equation}
\mathbf{G_2} = 
\begin{bmatrix}
1 & 0 \\
1 & 1
\end{bmatrix}=
\begin{bmatrix}
\mathbf{g}_0 \\
\mathbf{g}_1
\end{bmatrix},
\end{equation}
where $\mathbf{g}_{0}$
and $\mathbf{g}_{1}$ are the rows of $\mathbf{G_2}$. Hence,
\begin{equation}
\bGN=\mathbf{G}_2^{\otimes n}=
\begin{bmatrix}
1 & 0 \\
1 & 1
\end{bmatrix}^{\otimes n}.
\end{equation}
 


Lemma~\ref{lem:main} provides a criterion to determine the value of an entry in $\bGN$ based on the supports of the binary expansions of its row and column indices. An equivalent statement of the lemma was established earlier in~\cite{bardet} under the language of polynomials and their evaluations. 

\begin{lemma}\label{lem:main}
Let $g_{i,c}$ be the $(i,c)$-th entry of $\bGN$ for indices $i, c$ in $[0,N-1]$. Then the following holds. 
\begin{equation}
\label{eq:g_i_c}
g_{i,c} = \begin{dcases*}
        1, & \text{ if } $\mathcal{S}_c\subseteq \mathcal{S}_i$,\\
        0, & \text{ if } \text{otherwise}.
\end{dcases*}
\end{equation}
\end{lemma}
\begin{proof}
We show this by induction on $n$.

\textbf{Base case}. When $n=1$ and $N=2^1=2$, one can observe that (\ref{eq:g_i_c})  holds trivially. Note that supp$(0)=\varnothing$ and $\varnothing\subseteq\varnothing$ for entry (0,0), hence $g_{0,0}=1$.
\begin{equation}
\mathbf{G}_2=
\begin{pNiceArray}{cc}[first-row,first-col]
  & 0 & 1 \\
0 & 1 & 0 \\
1 & 1 & 1
\end{pNiceArray}.
\end{equation}

\textbf{Inductive step}. Suppose \eqref{eq:g_i_c} holds for $n$, we need to prove that it also holds for $n+1$. Let us use the notations $g_{i,c}^{n}$, $\mathcal{S}_i^{n}$, and $\mathcal{S}_c^{n}$ for $n$ and $g_{i,c}^{n+1}$, $\mathcal{S}_i^{n+1}$, and $\mathcal{S}_c^{n+1}$ for $n+1$. We have 

\begin{equation}\label{Pnplus1}
\mathbf{G}_{{2^{n+1}}}=
\begin{pNiceArray}{cc}
\mathbf{G}_{2^n} & \mathbf{0} \\
\mathbf{G}_{2^n} & \mathbf{G}_{2^n}
\end{pNiceArray}
\end{equation}

Now, we consider two cases:
\begin{itemize}
    \item \textbf{Case 1:} $0\leq i\leq2^n-1$ and $2^n\leq c\leq2^{n+1}-1$. 
    In this case we have $g_{i,c}=0$.
    Since $n\notin \mathcal{S}_i^{n+1}$ and $n\in \mathcal{S}_c^{n+1}$, we deduce that $\mathcal{S}^{n+1}_c \not\subseteq \mathcal{S}^{n+1}_i $. Thus, \eqref{eq:g_i_c} holds. 
    \item \textbf{Case 2:}  $2^n\leq i\leq2^{n+1}-1 $ or $0\leq c\leq2^{n}-1$.
    From (\ref{Pnplus1}) we have
    \begin{equation} \label{Old_Claim2}
    g_{i,c}^{n+1}=g_{(i\bmod 2^n),(c\bmod 2^n)}^n.
    \end{equation}
    
    \begin{claim} $\mathcal{S}_c^{n+1}\subseteq\mathcal{S}_i^{n+1} \iff \mathcal{S}^n_{c\bmod 2^n}\subseteq\mathcal{S}^n_{i\bmod 2^n}$.
    \end{claim}
    \begin{proof}
    Since
    \begin{equation}
    S_i^{n+1}=
    \begin{cases}
    \{n\}\cup\mathcal{S}_{i\bmod 2^n}^n, & \text{if } 2^n\leq i\leq2^{n+1}-1,\\
    \mathcal{S}_{i\bmod 2^n}^n=\mathcal{S}_{i}^n, & \text{if } 0\leq i\leq2^n-1,
    \end{cases}
    \end{equation}
    and
    \begin{equation}
    S_c^{n+1}=
    \begin{cases}
    \{n\}\cup\mathcal{S}_{c\bmod 2^n}^n, & \text{if } 2^n\leq c\leq2^{n+1}-1,\\
    \mathcal{S}_{c\bmod 2^n}^n=\mathcal{S}_{c}^n, & \text{if } 0\leq c\leq2^n-1.
    \end{cases}
    \end{equation}
    Under the assumption that $2^n\leq i\leq2^{n+1}-1 $ or $0\leq c\leq2^{n}-1$, if $n$ belongs to $S_c^{n+1}$ then it must also belong to $S_i^{n+1}$.
    From this, it is easy to see that Claim 1 holds.
    \end{proof}

    From Claim 1 and \eqref{Old_Claim2}, we can conclude that 
    \begin{equation}
    \begin{split}
    g_{i,c}^{n+1} &\underrel{\text{\eqref{Old_Claim2}}}{=} g_{(i\bmod 2^n),(c\bmod 2^n)}^n\\& \underrel{\text{Induction}}{=}
    \begin{dcases*}
        1 & \text{if }$\mathcal{S}_{c\bmod 2^n}^n\subseteq \mathcal{S}_{i\bmod 2^n}^n$,\\
        0 & \text{otherwise}\\
    \end{dcases*}
    \\&\underrel{\text{Claim 1}}{=} 
    \begin{dcases*}
        1, & \text{if }$\mathcal{S}_{c}^{n+1}\subseteq \mathcal{S}_{i}^{n+1}$,\\
        0, & \text{otherwise}\\
    \end{dcases*}
    \end{split}
    \end{equation}
\end{itemize}
Thus, the relation ($\ref{eq:g_i_c}$) holds for $n+1$ in Case 2 as well.\qedhere
\end{proof}

Lemma \ref{lem:main} is a fundamental tool that we rely on throughout this paper. 
Based on this lemma, we can easily determine the weight of a row $\bgi$ and the weight of a sum of two rows $\bgi+\bgj$ based on the supports of the binary representations of $i$ and $j$ as follows.

\begin{corollary}\label{cr:weight}
For $i$ and $j$ in $[0,2^n-1]$, we have
\[
\begin{split}
    \w(\bgi)&=2^{|\Si|},\\
    \w(\bgi+\bgj) &= 2^{|\Si|}+2^{|\Sj|} - 2\times 2^{|\Si\cap \Sj|}.
\end{split}
\]
\end{corollary}
\begin{proof}

From Lemma~\ref{lem:main}, we have
\[
\begin{split}
\supp(\bgi) &= \{c \in [0,N-1] \colon g_{i,c} = 1\}\\
&= \{c \in [0,N-1] \colon \Sc \subseteq \Si\},
\end{split}
\]
which implies that 
\[
\w(\bgi) = |\supp(\bgi)| = |\{c \in [0,N-1] \colon \Sc \subseteq \Si\}| = 2^{|\Si|}.
\]

We also have
\[
\supp(\bgi)\cap \supp(\bgj) = \{c \in [0,N-1] \colon \Sc \subseteq \Si \cap \Sj\},
\]
which implies that 
\[
|\supp(\bgi)\cap \supp(\bgj)| = 2^{|\Si\cap \Sj|}.
\]
Therefore,
\[
\begin{split}
\w(\bgi+\bgj) &= \w(\bgi) + \w(\bgj) - 2\times |\supp(\bgi)\cap \supp(\bgj)|\\
&= 2^{|\Si|} + 2^{|\Sj|} - 2\times 2^{|\Si\cap \Sj|},
\end{split}
\]
which proves the second equality.
\end{proof}

Recall that throughout this work we use the index subscript $i \in [0,2^n-1]$ and its $\Si = \supp(\bin(i))\subseteq [0,n-1]$ interchangeably. For instance, when $n = 5$, instead of $c_{10}$, we may write $c_{\{1,3\}}$ as $\S_{10}=\supp(01010) = \{1,3\}$.  

\begin{lemma} \label{lma:basic1}
For $i\in[0,2^n-1]$ and $\mathcal{J}\subseteq [i+1,2^n-1]$, let 
\begin{equation}
    \mathbf{c}=\bgi\oplus\bigoplus_{j\in\mathcal{J}}\bgj.
\end{equation}
Then for every subset $\S\subseteq\Si$ there exists at least one subset  $\T\subseteq \Ti \triangleq [0,n-1] \setminus \Si$  such that 
\begin{equation}
    c_{\S\cup\T}=1.
\end{equation}
\end{lemma}

\begin{proof}
For every $\S\subseteq\Si$, we define 
\begin{equation}\label{eq:Jprime}
    \J'(\S)\triangleq\{j\in\mathcal{J}:\S_j\supseteq\S\}.
\end{equation} 
We consider the following two cases.

\noindent {\bf Case 1: $|\J'(\S)|$ is even}. 
We pick $\T=\varnothing$, then $\S \cup \T = \S$, which is contained in both $\Si$ and $\Sj$. Therefore, $(\bgi)_{\S\cup\T}=(\bgj)_{\S\cup\T}=1$ for every $j\in \J'(\S)$, and $(\bgj)_{\S\cup\T}=0$ for every $j\in\mathcal{J}\setminus \J'(\S)$. Since $|\{i\}\cup\J'(\S)|$ is odd for every $\S\subseteq\Si$, we have $c_{\S\cup\T}=1$.

\noindent
{\bf Case 2: $|\J'(\S)|$ is odd.} 
We use double counting technique to count the elements of the set 
\[
\mathcal{P}=\{(\mathcal{T},j) \colon \mathcal{T}\subseteq\Sj\setminus\Si,\mathcal{T}\neq\varnothing,j\in\J'(\S)\}.
\]
First,
\[
|\mathcal{P}|=\sum_{j\in\J'(\S)} (2^{|\Sj\setminus\Si|}-1),
\]
which is odd as explained below. Note that the number of subsets of $\Sj\setminus\Si$ excluding the empty set is $2^{|\Sj\setminus\Si|}-1$, and $|\Sj\setminus\Si|\geq 1$ due to $\Sj\not\subseteq\Si$. On the other hand, $|\J'(\S)|$ is also odd. 
Therefore, $|P|$, which is the sum of an odd number of all odd terms, is odd.

Second, 
\[
|\mathcal{P}|=\sum_{\T\subseteq\Ti,\mathcal{T}\neq\varnothing} |\{j\in\J'(\S) \colon \T\subseteq\Sj\setminus\Si\}|.
\]
Therefore, since $|P|$ is odd, there exists at least one $\T\subseteq\Ti$ such that $|\{j\in \J'(\S) \colon \T \subseteq \Sj \setminus \Si\}|$ is odd. For this $\T$, we have $(\bgi)_{\S\cup\T}=0$ and 
\[
(\bgj)_{\S\cup\T}= \begin{dcases*}
        1, & \text{if } $j\in\J'(\S)\text{ and }\mathcal{T}\subseteq\Sj$,\\
        0, & \text{if } $j\in \J'(\S)\text{ and }\mathcal{T}\not\subseteq\Sj$,\\
        0, & \text{if } $j\in\mathcal{J}\setminus\J'(\S)$.\\
\end{dcases*}
\]
Therefore, $c_{\S\cup\T}=1$. 
\end{proof}


The following useful result, which was also established in~\cite[Corollary~1]{li2} via an induction proof, is a simple corollary of Lemma~\ref{lma:basic1}.

\begin{corollary}\label{cor:geq_wi}
For any $i\in[0,2^n-1]$ and $\mathcal{H}\subseteq [i+1,2^{n}-1]$, we have
\begin{equation}\label{eq:geq_wi}
    \w(\bgi\oplus\bigoplus_{h\in\mathcal{H}}\mathbf{g}_h)\geq \w(\bgi).
\end{equation}
\end{corollary}
\begin{proof}
According to Lemma \ref{lma:basic1}, for every subset $\S\subseteq\Si$, there exists at least one subset $\T\subseteq \Ti$ so that $c_{\S\cup\T}=1$. Since the total number of subsets of $\Si$ is $2^{|\Si|}$, we deduce that
\[
\w(\mathbf{c})= \w(\bgi\oplus\bigoplus_{h\in\mathcal{H}}\mathbf{g}_h)\geq 2^{|\Si|}\underset{\text{Lemma}~\ref{cr:weight}}{=}\w(\bgi). \qedhere
\]
\end{proof}


From Corollary \ref{cor:geq_wi}, the $d_{\min}$ of the code $\CI$ (including RM and polar codes) can be easily determined.
Different proofs of this result (via RM codes containing $\CI$) could be found in~\cite[Lemma~3]{hussami} (for polar codes) and \cite[Proposition~3]{bardet} (for decreasing monomial codes).

\begin{corollary} \label{cr:d_min}
The minimum distance of the code $\CI$ (see Section~\ref{sec:decomposition}), which includes RM and polar codes, is 
\begin{equation*}
    d_{\min}  = \min_{i\in\mathcal{I}} \w(\bgi),
\end{equation*}
where $\bgi$ denotes the $i$-th row in $\bGN$.
\end{corollary}
\begin{proof}
According to \eqref{eq:d_min_w_min}, the minimum distance of a linear code is the minimum-weight of any no-nzero codeword. 
From Corollary \ref{cor:geq_wi}, we know that the weight of every codeword in the coset $\CiI$, which has the form $\bgi\oplus\bigoplus_{j\in\mathcal{J}}\bgj$, is at least $\w(\bgi)$. Hence, it follows that $\dm = \wm  = \min_{i\in\mathcal{I}} \w(\bgi)$.
\end{proof}



\subsection{Proof of Theorem~\ref{thm:decomposition}}
\label{app:construction}

Assume that $\I\subseteq [0,N-1]$ satisfies the Partial Order Property, and $i\in \I$ satisfying $\w(\bgi)=\wm$. We first show in Lemma~\ref{lma:M_in_I} that for any $\varnothing \neq \J\subseteq \Ki$, the set $\MJ$ constructed by the $\M$-Construction satisfies $\MJ \subseteq \I\setminus [0,i]$. We then prove in Lemma~\ref{lem:Mwm} that $\MJ$ also satisfies \eqref{eq:decomposition}, that is, the sum of $\bgi$, $\bgj$ with $j\in \J$, and $\bgm$ with $m\in \MJ$ has weight $\wm$. These two lemmas together prove Theorem~\ref{thm:decomposition}. 

We need a simple auxiliary result to prove the lemma~\ref{lma:M_in_I}.

\begin{lemma}
\label{lem:extPO}
Suppose that for $i$ and $m$ in $[0,N-1]$, $\Si\setminus \Sm = \{a_1,\ldots,a_{\ell}\}$, $\Sm \setminus \Si = \{b_1,\ldots,b_u\}$, $\ell \leq u$, and moreover, $a_t<b_t$ for all $t \in [1,\ell]$. Then $i \preceq m$.
\end{lemma}
\begin{proof}
We define the indices $h_0,h_1,\ldots,h_{\ell+1}$ as follows.
\begin{itemize}
    \item $h_0 \triangleq i$,
    \item $\S_{h_t} \triangleq \big(\S_{h_{t-1}} \setminus \{a_t\}\big) \cup \{b_t\}$, $1\leq t \leq \ell$,
    \item $\S_{h_{\ell+1}} \triangleq \S_{h_{\ell}} \cup \{b_{\ell+1},\ldots,b_u\}$.
\end{itemize}
Clearly, $h_{\ell+1}=h_u$. Furthermore, by Definition~\ref{def:PO}, we have
\[
i = h_0 \preceq h_1 \preceq \cdots \preceq h_{\ell} \preceq h_{\ell+1} = m,
\]
which implies that $i \preceq m$.
\end{proof}

\begin{lemma} \label{lma:M_in_I} If $\I\subseteq [0,N-1]$ satisfies the Partial Order Property, $i\in \I$ satisfying $\w(\bgi)=\wm$, and $\J \subseteq \Ki$, then the set $\MJ$ generated by the $\M$-construction is a subset of $\I\setminus [0,i]$.
\end{lemma}
\begin{proof}
Since $\I$ satisfies the Partial Order Property, it suffices to show that $i\preceq m$ for every $m \in \MJ$. Note that $i\preceq m$ and $i\neq m$ imply $i < m$.

Take $m = \mJp$ created as in the $\M$-Construction.
Due to Lemma~\ref{lma:Ki}~(a), for each $j \in \Jp$, it holds that $|\Sj\setminus \Si|=1$ and either $|\Si\setminus \Sj| = 1$ or $|\Si\setminus \Sj|=0$ (i.e., $\Si \subseteq \Sj$).
Let $\Jp = \{j_1,j_2,\ldots,j_{|\Jp|}\}$. 
Rearranging, if necessary, let $1\leq \ell\leq p\leq u=|\Jp|$ be such that 
\begin{itemize}
    \item[(C1)] $|\Si\setminus \S_{j_1}|=|\Si\setminus \S_{j_2}|=\cdots=|\Si\setminus \S_{j_p}| = 1$,
    \item[(C2)] $|\Si\setminus \S_{j_{p+1}}|=|\Si\setminus \S_{j_{p+2}}|=\cdots=|\Si\setminus \S_{j_u}| = 0$,
     \item[(C3)] $\Si\setminus \S_{j_t}$, $t \in [1,\ell]$, are pair-wise disjoint (singleton) sets,
     \item[(C4)] $\cup_{t\in [\ell+1,p]}(\Si\setminus\S_{j_t}) \subseteq \cup_{t \in [1,\ell]}(\Si\setminus\S_{j_t})$.
\end{itemize}
Let $\Si\setminus \S_{j_t} = \{a_t\}$, $1\leq t\leq \ell$. Note that all these $\ell$ (singleton) sets are pair-wise disjoint according to (C3). From the $\M$-Construction, (C2), and (C4), we deduce that 
\[
\begin{split}
\Si \setminus \Sm &= \Si \setminus \big(\cap_{j \in \Jp}\Sj\big)\\
&= \cup_{j\in \Jp} \big(\Si\setminus\Sj)
= \cup_{t \in [1,u]} \big(\Si\setminus\S_{j_t})\\
&= \cup_{t \in [1,\ell]} \big(\Si\setminus\S_{j_t}) = \{a_1,\ldots,a_{\ell}\},
\end{split}
\]
where the second equality is due to De Morgan's laws and the fourth is due to (C2) and (C4). 
Let $\S_{j_t}\setminus \Si = \{b_t\}$, $1 \leq t \leq u$. Then due to the $\M$-Construction, all these (singleton) sets are pair-wise disjoint as well. 
Moreover, as $i \preceq j_t$, we have $a_t < b_t$ for all $1 \leq t \leq \ell$. 
We also have
\[
\begin{split}
\Sm \setminus \Si &= \cup_{j \in \Jp} (\Sj\setminus \Si)\\
&= \cup_{t \in [1,u]} (\S_{j_t}\setminus \Si)=\{b_1,\ldots,b_u\}.
\end{split}
\]
Applying Lemma~\ref{lem:extPO} to $\Si$ and $\Sm$, we conclude that $i \preceq m$ as desired.  
\end{proof}

Before proving our key Lemma~\ref{lem:Mwm}, we need the following important result.

\begin{lemma}\label{lma:SUT_0_1}
If $\I\subseteq [0,N-1]$ satisfies the Partial Order Property, $i\in \I$ satisfies $\w(\bgi)=\wm$, and $\J \subseteq \Ki$, then
for each subset $\mathcal{S}\subseteq \Si$, there exists a unique subset $\mathcal{T}^*(\S)\subseteq \R$, where $\R \triangleq \cup_{j\in\mathcal{J}}\big(\Sj\setminus\Si\big)$, such that
\begin{equation*}
\label{eq:SUT_0_1}
c_{\mathcal{S}\cup \T} \triangleq \begin{dcases*}
        1, & \text{if } $\mathcal{T} = \mathcal{T}^*(S)$,\\
        0, & otherwise,
\end{dcases*}
\end{equation*}
where 
\[
\mathbf{c} \triangleq \mathbf{g}_{i} \oplus \bigoplus_{j \in \mathcal{J}} \mathbf{g}_{j} \oplus \bigoplus_{m \in \mathcal{M}} \mathbf{g}_{m}.
\]
and $c_{\mathcal{S}\cup \T}$ denotes a coordinate of $\mathbf{c}$ indexed by ${\S\cup \T}$. 
\end{lemma}
\begin{proof}
We prove the lemma by providing an explicit construction of the set $\mathcal{T}^*(S)$ for every $\S\subseteq \Si$. First, let $\J^*(\S)$ denote the set of rows in $\mathcal{J}$ such that $\S\subseteq \Sj$. 
\begin{equation*}
\mathcal{J}^*(\mathcal{S})=\{j\in\mathcal{J} \colon \S\subseteq\Sj\}.
\end{equation*}
We define $\T^*(\S)$ as the set consisting of indices in $\R = \cup_{j\in\mathcal{J}}\big(\Sj\setminus\Si\big)$ that belong to an odd number of $\Sj$, $j\in \J^*(\S)$.

We divide the remainder of the proof into two parts, showing that $c_{\S\cup \T} = 1$ if $\T = \T^*(\S)$ in Lemma \ref{lma:c=1} and $c_{\S\cup \T} = 0$ if $\T \neq \T^*(\S)$ in Lemma \ref{lma:c=0}. Both lemmas can be found at the end of this appendix. 
\end{proof}

We illustrate in Example~\ref{ex:T*} how $\T^*(\S)$ discussed in the proof of Lemma~\ref{lma:SUT_0_1} can be found.

\begin{example}
\label{ex:T*}
We consider $n = 4$, $N = 16$, $i = 3$, and $\J$ as given in Example \ref{ex:MJ1}: 
\[
\begin{split}
\mathcal{J} &=\{5,6,7,9,10\}\\ 
&=\{(0101)_2,(0110)_2,(0111)_2,(1001)_2,(1010)_2\}.
\end{split}
\]
Note that $\R = \cup_{j\in\mathcal{J}}\big(\Sj\setminus\Si\big) = \{2,3\}$.
The subset $\mathcal{S}\subseteq\Si=\{0,1\}$ can be $\varnothing, \{0\}$, $\{1\}$, or $\{0,1\}$.
Let us consider $\mathcal{S}=\{{\color{blue}1}\}$. Then, 
\[
\mathcal{J}^*(S)= \{6,7,10\} = \{(01{\color{blue}1}1)_2,(01{\color{blue}1}0)_2,({\color{red}1}0{\color{blue}1}0)_2\}.
\]
Since 2 appears twice and 3 appears once among $\S_6$, $\S_7$, and $\S_{10}$, we have $\mathcal{T}^*(\S)= \{{\color{red}3}\}$. Consequently, $c_{\mathcal{S}\cup \T^*(\S)}=c_{(1010)_2}=c_{10}=1$. 

To determine all the `1' coordinates in the codeword, we find the sets $\T^*(\S)$ corresponding to other subsets $\mathcal{S}$ of $\Si$, which are $\mathcal{T}^*(\{0\})=\{3\}$, $\mathcal{T}^*(\{0,1\})=\{2\}$, and $\mathcal{T}^*(\varnothing)=\{2\}$. Then, $c_{\mathcal{S}\cup \T^*(\S)}=c_{(1001)_2}=c_{9}=1$, $c_{(0111)_2}=c_{7}=1$, and $c_{(0100)_2}=c_{4}=1$. The remaining coordinates in the codeword are `0'. 
\end{example}

\begin{lemma} \label{lem:Mwm} If $\I\subseteq [0,N-1]$ satisfies the Partial Order Property, $i\in \I$ satisfies $\w(\bgi)=\wm$, and $\J \subseteq \Ki$, then the set $\MJ$ generated by the $\M$-construction satisfies 
\[
\w\big(\bgi\oplus \bigoplus_{j\in\J}\bgj \oplus  \bigoplus_{m\in\M(\J)}\bgm\big) = \wm.
\]
\end{lemma}
\begin{proof} 
Let $\bc = \bgi\oplus \bigoplus_{j\in\J}\bgj \oplus  \bigoplus_{m\in\M(\J)}\bgm$. 
Recall that we use the set $\S_h$, $h \in [0,N-1]$, to index the coordinate $c_h$ of $\bc$, where $\S_h\triangleq \supp(\bin(h))\subseteq [0,n-1]$. Due to Lemma~\ref{lem:main}, $\bgi$ and $\bgj$, $j\in \J$, have a zero entry at every index $h$ with $\S_h \not\subset \Si \cup \R$, where $\R\triangleq \big(\cup_{j\in\mathcal{J}}\big(\Sj\setminus\Si\big)\big)$. Moreover, due to the way we construct $\MJ$ in the $\M$-Construction, since $\Sm \subseteq \Si \cup \R$, the row $\bgm$ with $m\in \MJ$ also has a zero at such indices. Therefore, to determine the Hamming weight of $\bc$, we only need to consider the coordinates of $\bc$ indexed by subsets of $\Si\cup \R$ and ignore the rest, because they are all zeros.


Let $\T^*(\S)$ be defined as in the statement of Lemma~\ref{lma:SUT_0_1} for each set $\S \subseteq \Si$. Since $\Si \cap \R =\varnothing$, the collection of subsets of $\Si \cup \R$ can be written as
\begin{multline*}
\{\mathcal{S}\cup\mathcal{T} \colon   \mathcal{S}\subseteq\mathcal{S}_i, \mathcal{T}\subseteq\R\} = \\ 
\{\mathcal{S}\cup\T^*(\S) \colon \S \subseteq \Si\}\cup\{\mathcal{S}\cup\mathcal{T} \colon \S \subseteq \Si, \mathcal{T}\neq\T^*(\S)\}.
\end{multline*}
According to Lemma~\ref{lma:SUT_0_1}, $c_{\S\cup\T}=1$ if $\T = \T^*(\S)$ and $0$ otherwise.
Therefore,
\[
\begin{split}
\w(\mathbf{c})&=|\{\mathcal{S}\cup\T^*(\S) \colon    \mathcal{S}\subseteq\mathcal{S}_i\}| \\ 
&=|\{\mathcal{S}\subseteq\mathcal{S}_i\}|=2^{|\mathcal{S}_i|}=\wm.
\end{split}
\]
This proves the lemma.
\end{proof}

The next two lemmas settle the remaining parts in the proof of Lemma~\ref{lma:SUT_0_1}.

\begin{lemma}\label{lma:c=1}
If $\I\subseteq [0,N-1]$ satisfies the Partial Order Property, $i\in \I$ satisfies $\w(\bgi)=\wm$, $\J \subseteq \Ki$, and $\MJ$ is created by the $\M$-Construction, then for every $\mathcal{S}\subseteq\mathcal{S}_i$,
\[
c_{\mathcal{S}\cup \mathcal{T}^*(\S)} = \bigoplus_{m\in \{i\}\cup \mathcal{J} \cup \MJ} g_{m,\mathcal{S}\cup \T^*(\S)}= 1.
\]
where $\T^*(\S)$ is defined as in the proof of Lemma~\ref{lma:SUT_0_1} and $c_{\mathcal{S}\cup \mathcal{T}^*(\S)}$ is the coordinate indexed by $\mathcal{S}\cup \mathcal{T}^*(\S)$ of
\[
\mathbf{c}\triangleq\mathbf{g}_{i} \oplus \bigoplus_{j \in \mathcal{J}} \mathbf{g}_{j} \oplus \bigoplus_{m \in \mathcal{M}} \mathbf{g}_{m}.
\]
\end{lemma}
\begin{proof}
Recall that from Lemma \ref{lem:main}, $\bg_{m,c}=1$ if and only if $\S_c\subseteq\S_m$. To prove Lemma \ref{lma:c=1}, it suffices to show that the number of $m \in \{i\}\cup \mathcal{J} \cup \mathcal{M}$ satisfying $\mathcal{S}\cup \T^*(\S) \subseteq \S_m$ is odd (so that $c_{\mathcal{S}\cup \mathcal{T}^*(\S)}=1$). 

Note that in the $\M$-Construction, the rows of $\MJ$ are $m_{\Jp}$ with $|\Jp|\geq 2$. 
To facilitate the proof, we also include $\Jp$ with $|\Jp| < 2$ by setting $\S_{m_{\varnothing}} \triangleq \S_i$ and $\S_{m_{\{j\}}} \triangleq \Sj$. In this way, any row index $m \in \{i\}\cup \J \cup \MJ$ corresponds to an element $m_{\Jp}$ for some $\Jp \subseteq \J$. 

According to the definition of $\J^*(\S)$ in the first part of the proof of Lemma~\ref{lma:SUT_0_1},
for any $\J'\not\subseteq\mathcal{J}^*(\mathcal{S})$, we have $\S \not\subseteq \Sj$ for some $j\in \Jp$, and hence, $\mathcal{S}\not\subseteq\Si\cap\big(\cap_{j\in\J'}\Sj\big)$. 
Furthermore, since $\S \subseteq \Si$, we have $\S \not\subseteq \cup_{j\in \Jp}(\Sj\setminus \Si)$. 
Therefore, 
\[
\S\cup\T^*(\S) \not\subseteq \cup_{j\in \Jp}(\Sj\setminus \Si) \cup \big(\Si\cap\big(\cap_{j\in\Jp}\Sj\big)\big) =\mathcal{S}_{m_{\J'}}.
\]
Thus, we only need to consider $m=m_{\J'}$ where $\Jp\subseteq\mathcal{J}^*(\mathcal{S})$. 
Note that if $\Jp \subseteq \J^*(\S)$ then $\S \subseteq \Si\cap\big(\cap_{j\in\Jp}\Sj\big) \subseteq \S_{m_{\Jp}}$.
To have $\S \cup \T^*(\S)\subseteq \S_{m_{\Jp}}$, we only need $\T^*(\S) \subseteq \cup_{j\in \Jp}(\Sj\setminus \Si)$.
According to the $\M$-Construction, the sets $\Jp$ of interest should also satisfy that the sets $\Sj\setminus \Si$, $j \in \Jp$, are pairwise disjoint. Note that we now can safely ignore the requirement that $|\Jp| \geq 2$ in the $\M$-Construction as we have set a convention for the cases $\Jp=\varnothing$ and $\Jp=\{j\}$.

From the above discussion, it suffices to show that the number of sets $\Jp\subseteq \J^*(\S)$ satisfying the following conditions (C1) and (C2) is odd.
\begin{itemize}
    \item (C1) $\Sj\setminus \Si$, $j \in \Jp$, are pairwise disjoint.
    \item (C2) $\T^*(\S) \subseteq \cup_{j\in \Jp}(\Sj\setminus \Si)$.
\end{itemize}
To this end, for each $j\in \JsS$ let $r_j$ be the unique element in $\Sj\setminus \Si$, and for each $r\in \R \triangleq \cup_{j\in\mathcal{J}}\big(\Sj\setminus\Si\big)$ denote $\J_r\triangleq \{j\in \JsS\colon r_j=r\}$. Moreover, let $\O\triangleq \{r \in \R\colon |\J_r|\text{ is odd}\}$ and $\E\triangleq \{r \in \R\colon |\J_r|\text{ is even}\}$. Then $\T^*(\S)=\O$ and $\R=\O\cup \E$.

Recall that we aim to prove that the number of $\Jp\subseteq \JsS$ satisfying (C1) and (C2) is odd. We note that these conditions are satisfied if and only if 
\[
|\J'\cap\J_r|= 
\begin{dcases*}
        1, & \text{for} $r \in \O$,\\
        0 \text{ or } 1, & \text{for} $r \in \E$.
\end{dcases*}
\]
In fact, to satisfy (C1), $\Jp$ must contain \textit{at most} one element from each $\J_r$, for every $r \in \R$. To also satisfy (C2), because $\T^*(\S)=\O$, $\Jp$ must contain \textit{exactly} one element from each $\J_r$ for every $r\in \O$.
Note that there are $|\J_r|$ ways to pick an element from $\J_r$, $r\in \O$, and $|\J_r|+1$ ways to choose no element or one from $\J_r$, $r \in \E$.
Thus, the number of $\Jp\subseteq \JsS$ that meets both (C1) and (C2) is equal to 
\[
\prod_{r\in \O} \underbrace{|\mathcal{J}_r|}_{\text{Odd}}\times \prod_{r \in \E} \underbrace{(\underbrace{|\mathcal{J}_r|}_{\text{Even}}+1)}_{\text{Odd}},
\]
which is an odd number as desired.  
\end{proof}

\begin{lemma}\label{lma:c=0}
If $\I\subseteq [0,N-1]$ satisfies the Partial Order Property, $i\in \I$ satisfies $\w(\bgi)=\wm$, $\J \subseteq \Ki$, and $\MJ$ is created by the $\M$-Construction, then for every $\mathcal{S}\subseteq\mathcal{S}_i$ and $\mathcal{T}\neq \mathcal{T}^*(\S)$,
\[
c_{\mathcal{S}\cup T} = \bigoplus_{m\in \{i\}\cup \mathcal{J} \cup \MJ} g_{m,\mathcal{S}\cup \mathcal{T}}= 0,
\]
where $c_{\mathcal{S}\cup T}$ is the coordinate indexed by $\S\cup\T$ of 
\[
\mathbf{c}\triangleq\mathbf{g}_{i} \oplus \bigoplus_{j \in \mathcal{J}} \mathbf{g}_{j} \oplus \bigoplus_{m \in \mathcal{M}} \mathbf{g}_{m}.
\]
\end{lemma}
\begin{proof}
We follow the same proof strategy as in Lemma~\ref{lma:c=1} and also define the sets $\J_r$, $\O$, and $\E$, noting that $\R=\O\cup \E$ and $\T^*(\S)=\O$. 
For all $\S\subseteq \Si$ and $\T\neq \T^*(\S)$, our objective is to show that the number of $\Jp\subseteq \JsS$ satisfying both (C1) and (C3) is even (so that $c_{\S\cup\T}=0$), with
\begin{itemize}
    \item (C1) $\Sj\setminus \Si$, $j \in \Jp$, are pairwise disjoint.
    \item (C3) $\T \subseteq \cup_{j\in \Jp}(\Sj\setminus \Si)$.
\end{itemize}
We observe that a set $\Jp\subseteq \JsS$ satisfies (C1) if and only if $\Jp$ contains \textit{at most} one element of each $\J_r$, for every $r\in \R$. Furthermore, the set $\Jp$ also satisfies (C3) if and only if it contains \textit{exactly} one element of each $\J_r$ with $r \in \T$. Therefore, the number of $\Jp\subseteq \JsS$ that meets both (C1) and (C3) is
\begin{equation}
\label{eq:number_Jp}
\prod_{r\in \T} |\J_r|\times \prod_{r \in \R\setminus \T}(|\J_r|+1),
\end{equation}
which is an even number. Indeed, as $\T \neq \T^*(\S)=\O$, either $\T\cap \E \neq \varnothing$ or $\R \setminus \T \supseteq \O\setminus \T \neq \varnothing$. 
If the former holds, then the first product in \eqref{eq:number_Jp} contains an even factor $|\J_r|$ for some $r \in \E\cap \T$ and is therefore even. If the latter holds, then the second product contains an even factor $(|\J_r|+1)$ for some $r \in \O\setminus \T$. In either case, the number of $\Jp$ as given by \eqref{eq:number_Jp} is even as desired
\end{proof}

\end{document}